\documentclass[10pt,twocolumn,journal,letterpaper]{IEEEtran}%
\usepackage[square,comma,sort&compress,numbers]{natbib}
\usepackage{algpseudocode}
\usepackage{amsmath}
\usepackage{amsfonts}
\usepackage{amssymb}
\usepackage{amsthm}
\usepackage{mathtools}
\usepackage[detect-weight=true,detect-mode=true]{siunitx}
\usepackage{nicefrac}
\usepackage{dsfont}
\usepackage{upgreek}
\usepackage{bbm}
\usepackage{booktabs}
\usepackage[hidelinks]{hyperref}%

\theoremstyle{problemStyle}
\newtheorem{theorem}{Theorem}
\newtheorem{lemma}{Lemma}
\newtheorem{definition}{Definition}
\newtheorem{assumption}{Assumption}

\algrenewcommand\algorithmicrequire{\textbf{Input:}}
\algrenewcommand\algorithmicensure{\textbf{Output:}}

\algnewcommand\algorithmicinit{\textbf{Initialization:}}
\algnewcommand\Init{\item[\algorithmicinit]}

\algnewcommand\algorithmicparam{\textbf{Parameters that need to be selected:}}
\algnewcommand\Param{\item[\algorithmicparam]}

\algnewcommand\algorithmicproc{\textbf{Procedure:}}
\algnewcommand\Proc{\item[\algorithmicproc]}

\newcommand{\argmin}[1]{\underset{#1}{\operatorname{arg\,min\,}}}
\newcommand{\argmax}[1]{\underset{#1}{\operatorname{arg\,max\,}}}

\makeatletter
\newcommand{\pushright}[1]{\ifmeasuring@#1\else\omit\hfill$\displaystyle#1$\fi\ignorespaces}
\newcommand{\pushleft}[1]{\ifmeasuring@#1\else\omit$\displaystyle#1$\hfill\fi\ignorespaces}
\makeatother

\begin{document}

\title{High Precision TOA-based Direct Localization of Multiple Sources in Multipath}
\author{
	Nil~Garcia*\dag,
	Alexander~M.~Haimovich*,
	Martial~Coulon$\dag$
	and Jason~A.~Dabin** \\
	* CWCSPR, New Jersey Institute of Technology, USA, \texttt{nil.garcia@njit.edu, haimovich@njit.edu}\\
	$\dag$ University of Toulouse, INP-ENSEEIHT/IRIT, France, \texttt{martial.coulon@enseeiht.fr}\\
	** SPAWAR Systems Center Pacific, San Diego (CA), USA 
	\thanks{The work of A.\ M.\ Haimovich was partially supported by the U.S.\ Air Force Office of Scientific Research under agreement No.\ FA9550-12-1-0409.}
	\thanks{This paper was presented in part at 2014 48th Asilomar Conference on Signals, Systems and Computers.}
}

\maketitle

\begin{abstract}

Localization of radio frequency sources over multipath channels is a difficult problem arising in applications such as outdoor or indoor gelocation. Common approaches that combine ad-hoc  methods for multipath mitigation with indirect localization relying on intermediary parameters such as time-of-arrivals, time difference of arrivals or received signal strengths, provide limited performance. This work models the localization of known waveforms over unknown multipath channels in a sparse framework, and develops a direct approach in which multiple sources are localized jointly, directly from observations obtained at distributed sources. The proposed approach exploits channel properties that enable to distinguish line-of-sight (LOS) from non-LOS signal paths. Theoretical guarantees are established for correct recovery of the sources' locations by atomic norm minimization. A second-order cone-based algorithm is developed to produce the optimal atomic decomposition, and it is shown to produce high accuracy location estimates over complex scenes, in which sources are subject to diverse multipath conditions, including lack of LOS.

\end{abstract}

\section{Introduction}

Traditional time-of-arrival (TOA)-based localization is accomplished through a two-step process. In the first step, sensors estimate TOA's from all incoming signals; in the second step, such estimates are transmitted to a central node, that subsequently estimates the location of each source by multilateration \cite{Guvenc09}.  We refer to these localization techniques as \textit{indirect} localization.
In a multipath environment, each sensor receives, in addition to a line-of-sight (LOS) signal, multiple (possibly overlapping) replicas due to non-line-of-sight (NLOS) paths. Due to these multiple arrivals, it is, in general, more challenging to obtain accurate TOA estimates of the LOS components at the sensors. Matched filtering is a method for time delay estimation. However, its performance degrades greatly in the presence of multipath whose delay is of the same order than the inverse of the bandwidth of the signal \cite{Dardari08}.
Moreover, in the case of blockage of the LOS path, the TOA of the first arrival does not correspond to a LOS component anymore, and will corrupt localization. In such a case, it is customary to apply techniques, like the one in \cite{Chen99}, to mitigate NLOS channel biasing of the geolocation estimate.

A better approach than indirect localization is to infer the source locations directly from the signal measurements without estimating any parameters such as propagation delays. 
The concept of direct localization was first introduced by Wax and Kailath \cite{Wax82,Wax83} in the 70's. However, it is in the last decade that  Weiss et al.\ have further investigated and proposed actually efficient techniques \cite{Weiss05,Picard10,Bar-Shalom11,Bialer13}. In the absence of multipath, the state-of-the-art is  Direct Position Determination (DPD) \cite{Weiss05} which outperforms standard indirect localization, particularly at low signal-to-noise ratio (SNR), because it takes into account the fact that signals arriving at different sensors are emitted from the same location. 
The literature on direct localization in the presence of multipath is scarce. In \cite{Papakonstantinou08} a maximum likelihood (ML) estimator has been developed for the location of a single source assuming a fixed and known number of multipath, but without providing an efficient way to compute the estimator. In \cite{Bialer13}, a Direct Positioning Estimation (DPE) technique is proposed for operating in dense multipath environments, but requires knowledge of the power delay profile and is limited to localization of a single source.

A requirement of direct localization is that the signals, or a function of them, are sent to a fusion center which estimates the source's locations. Thus direct techniques are best suited for centralized networks. An example of this are Cloud Radio Access Networks (C-RAN) \cite{ChinaMobile11,Wu12}. C-RAN is a novel architecture for wireless cellular systems whereby the base stations relay the received signals to a central unit which performs all the baseband processing. Cellular systems are required to be location-aware, that is they must be able to estimate the locations of the user equipments (UE) for applications such as security, disaster response, emergency relief and surveillance in GPS-denied environment \cite{Gezici05}. In addition, in the
USA, it is required by the Federal Communications Commission (FCC) that by 2021 the wireless service providers must locate UE's initiating an Emergency 911 (E-911) call with an accuracy of 50 meters for 80\% of all wireless E-911 calls \cite{FCC911}. In uplink localization, the base stations perform time measurements of the received signals emitted by the UE's in order to infer their positions. Thus, a high accuracy direct localization technique designed for multipath channels, such as the one proposed in this work, may enhance the localization accuracy of current existing cellular networks by utilizing the C-RAN infrastructure. Moreover, it exists other applications that may benefit from high accuracy TOA-based geolocation such as in WLAN and WPAN networks. For instance, the setup of \cite{Flammini13} uses radios with the IEEE 802.15 (WPAN) standard to localize devices. The setups proposed in \cite{Golden07,Ciurana07} employ TOA-based localization for localizing 802.11 devices (WLAN). In \cite{Pahlavan06} it is proposed  a hybrid RSS(received signal strength)-TOA based localization algorithm that works with 802.11 and 802.15 technologies. Other TOA-based localization applications are in the radio frequency identification (RFID) field \cite{Arnitz11}.
  
In this paper, we present a TOA-based \textbf{d}irect \textbf{l}ocalization technique for multiple sources in \textbf{m}ultipath environments (DLM) assuming known waveforms and no prior knowledge of the statistics of the NLOS paths. Preliminary results of DLM were presented in \cite{Garcia14}. Without some prior knowledge on the multipath, NLOS components carry no information, and the best performance is obtained by using only LOS components \cite{Gezici05}. We propose an innovative approach, based on ideas of compressive sensing and atomic norm minimization \cite{Chandrasekaran12}, for jointly estimating the sources' locations using as inputs the signals received at all sensors. Numerical evidence shows that DLM has higher accuracy than indirect techniques, and that it works well in a wide range of multipath scenarios, including sensors with blocked LOS. Moreover DLM requires no channel state information.

In Section~\ref{sec:signal}, we introduce the signal model. Section~\ref{sec:DLM} briefly introduces our proposed technique. Sections \ref{sec:1stPhase} and \ref{sec:2ndPhase} provide in-depth explanations on the different parts of our technique. Section~\ref{sec:simulations} compares DLM to previous
existing techniques. Finally, Section~\ref{sec:conclusions} reports our conclusions.

\section{Signal Model} \label{sec:signal}

Consider a network composed of $L$ sensors and $Q$ sources located in a plane. The location of the $q$-th source is defined by two coordinates stacked in a vector $\mathbf{p}_q$. All sources share the same bandwidth $B$, and transmit their own signals $\{s_q(t)\}_{q=1}^Q$. The number of sources $Q$ and their waveforms are known. The observation time is $T$, assumed to be shorter than the time coherence of the channel, therefore, the channel is time-invariant. The complex-valued baseband signal at the $l$-th sensor is
\begin{equation} \label{eq:signal_received}
	r_l(t) = r_l^{\text{LOS}}(t) + r_l^{\text{NLOS}}(t) +w_l(t)\qquad 0\leq t \leq T,
\end{equation}
where $w_l(t)$ is circularly-symmetric complex white Gaussian noise with known variance $\operatorname{E}|w_l(t)|^2 = \sigma_w^2$. The term $r_l^{\text{LOS}}(t)$ is the sum of all LOS components:
\begin{equation} \label{eq:signal_LOS}
	r_l^{\text{LOS}}(t) = \sum_{q=1}^{Q}\alpha_{ql}s_q\left(t-\tau_l(\mathbf{p}_q)\right),
\end{equation}
where $\alpha_{ql}$ is an unknown complex scalar representing the signal strength and phase of the LOS path between the $q$-th source and $l$-th sensor, and $\tau_l(\mathbf{p})$ is the delay of a signal originating at $\mathbf{p}$ and reaching the $l$-th sensor:
\begin{equation} \label{eq:delay_function}
	\tau_l(\mathbf{p})=\left\|\mathbf{p}-\mathbf{p}_l'\right\|_2/c.
\end{equation}
In \eqref{eq:delay_function}, $\mathbf{p}_l'$ is the location of the $l$-th sensor, $c$ is the speed of light and $\|\cdot\|_2$ denotes the standard Euclidean norm. The term $r_l^{\text{NLOS}}(t)$ in \eqref{eq:signal_received} aggregates all NLOS arrivals:
\begin{equation} \label{eq:signal_NLOS}
	r_l^{\text{NLOS}}(t) = \sum_{q=1}^{Q} \sum_{m=1}^{M_{ql}}\alpha_{ql}^{(m)} s_q\left(t-\tau_{ql}^{(m)}\right),
\end{equation}
where $M_{ql}$ denotes the unknown number of NLOS paths between the $q$-th source and the $l$-th sensor, $\alpha_{ql}^{(m)}$ is an unknown complex scalar representing the amplitude of the $m$-th NLOS path between the $q$-th source and $l$-th sensor, and $\tau_{ql}^{(m)}$ is the delay of the NLOS component. The received signal \eqref{eq:signal_received} is sampled at a frequency $f_s$ satisfying the Nyquist sampling criterion: $f_s\geq 2B$, where $B$ is the bandwidth of $r(t)$.  Each sensors collects $N$ time samples at each observation time. By stacking the $N$ acquired samples, the received signal $\mathbf{r}_l = [r_l(0),\ldots,r_l((N-1)/f_s)]^T$ at the $l$-th sensor can be written in the following vector form
\begin{equation} \label{eq:signal_sampled}
	\mathbf{r}_l =
	\sum_{q=1}^{Q}\alpha_{ql}\mathbf{s}_q\left(\tau_l(\mathbf{p}_q)\right)+
	\sum_{q=1}^{Q}\sum_{m=1}^{M_{ql}}\alpha_{ql}^{(m)}\mathbf{s}_q\left(\tau_{ql}^{(m)}\right)
	+\mathbf{w}_l,
\end{equation}
where $\mathbf{s}_q(\tau)$ is the vector of the $N$ received samples from the $q$-th source waveform with delay $\tau$:
\begin{equation} \label{eq:signal_delayed}
	\mathbf{s}_q(\tau)=
	\begin{bmatrix}
		s_q\left(0-\tau\right) &
		\cdots &
		\ s_q\left((N-1)/f_s-\tau\right)
	\end{bmatrix}^T.
\end{equation}
Since all sensors acquire the same number of samples, the samples may be stacked in an $N\times L$ matrix 
\begin{multline} \label{eq:signal_matrix}
	\mathbf{R} =
	\begin{bmatrix}
		\mathbf{r}_1 & \cdots & \mathbf{r}_L
	\end{bmatrix}= \\
	=\sum_{q=1}^{Q}
	\begin{bmatrix}
		\alpha_{q1}\mathbf{s}_q\left(\tau_1(\mathbf{p}_q)\right)
		& \cdots &
		\alpha_{qL}\mathbf{s}_q\left(\tau_L(\mathbf{p}_q)\right)
	\end{bmatrix}
	+\\
	+\sum_{q=1}^{Q}\sum_{l=1}^{L}\sum_{m=1}^{M_{ql}}
	\alpha_{ql}^{(m)}\mathbf{s}_q\left(\tau_{ql}^{(m)}\right)\mathbf{v}_l^T
	+\mathbf{W},
\end{multline}
where the rows and columns index time instants and sensors, respectively, and $\mathbf{v}_l$ is an all-zeros vector except for the $l$-th entry which is one.
The LOS and NLOS components are parametrized by the first and second summands in \eqref{eq:signal_matrix}, respectively. In the rest of the paper, we will switch between the notations in \eqref{eq:signal_sampled} and \eqref{eq:signal_matrix} depending on whether we are interested in the signal of one sensor only or of all sensors.

\section{Proposed Localization Technique} \label{sec:DLM}

In order to develop a localization technique, it is first necessary to understand what parameters of the received signals depend on the sources locations. 
In the signal model introduced in the previous section, the propagation delays of the NLOS components \eqref{eq:signal_NLOS} were assumed to be unknown and arbitrary, because of the lack of prior statistical knowledge of the channel. Thus, information on the sources locations is carried only by the LOS components \eqref{eq:signal_LOS}. This claim is supported by the analysis in \cite{Gezici05}, which showed that the CRB increases when NLOS components are present. Consequently, without a priori knowledge, the optimal strategy is to reject NLOS components as much as possible, and rely on the LOS components to infer the sources' locations.

With indirect techniques, first, the TOA's of the LOS components are estimated, and then used to localize the sources by multilateration. However, indirect techniques are suboptimal because they estimate the TOA of the first path at each sensor independently, instead of taking into account that all LOS components originate from a single source location. In this section, we propose a direct localization technique that relies on the fact that all LOS components associated with a source must originate from the same location. Under the Gaussian assumption, the maximum likelihood estimator (MLE) is the solution to the following fitting problem
\begin{multline} \label{problem:ML}
	\min_{\substack{
	\mathbf{p}_1,\ldots,\mathbf{p}_Q\\
	\alpha_{11},\ldots,\alpha_{LQ}\\
	M_{11},\ldots,M_{LQ}\\
	\tau_{11}^{(1)},\ldots,\tau_{LQ}^{(M_{LQ})}\\
	\alpha_{11}^1,\ldots,\alpha_{LQ}^{M_{LQ}}
	}}
	\sum_{l=1}^{L}\Bigg\|
	\mathbf{r}_l
	-\sum_{q=1}^{Q}\alpha_{ql}\mathbf{s}_q\left(\tau_l(\mathbf{p}_q)\right)
	-\\[-8ex]
	-\sum_{q=1}^{Q}\sum_{m=1}^{M_{ql}}\alpha_{ql}^{(m)}\mathbf{s}_q\left(\tau_{ql}^{(m)}\right)
	\Bigg\|_2^2
\end{multline}
subject to $\tau_{ql}^{(m)}>\tau_l(\mathbf{p}_q)$, for all $q$, $l$ and $m$. The parameters of interest are the source locations $\{\mathbf{p}_q\}_{q=1}^Q$, while the rest act as nuisance parameters. Besides the fact that it is an enormous challenge to find an efficient technique for minimizing this objective function, the ML criterion does not even lead to a satisfactory solution. The reason is that $M_{ql}$, for all $l$ and $q$, are hyperparameters that control the number of NLOS paths in our model. It is known that increasing the values of hyperparameters always leads to a better fitting error \cite{Koller}, and in our case, it would lead to the erroneous conclusion that there are a very large number of NLOS arrivals. Instead, we assume that the number of NLOS arrivals and the number of sources is low with respect to the number of observations. This assumption enables the formulation of a feasible solution to the ML multipath estimation problem by means of a sparse recovery technique.

In order to obtain a high-precision localization technique, there are two properties of the signal paths that need to be exploited. These properties allow to distinguish LOS from NLOS components. The first one is that NLOS components arrive with a longer delay than LOS components, and the second property is that all LOS paths originate from the same location.
Our technique is divided into two stages, which are explained in the following two sections. In the first stage, NLOS components are canceled out from the received signals by exploiting the fact that LOS components must arrive first. This processing can be done locally at each sensor. In the second stage, the cleaned version of the received signals are sent to a fusion center that finds the sources' location. It is in this stage that the source locations are estimated by exploiting the fact that LOS components must originate from the same location, whereas NLOS components may be local to the sensors.

\section{Stage 1: Deconvolution} \label{sec:1stPhase}

In this stage, the multipath channel is deconvolved, or equivalently, the propagation delays of different paths are estimated, and the multipath contributions are removed from the received signals. 
Our technique of choice for deconvolution  is the sparsity-based delay estimation technique proposed by Fuchs \cite{Fuchs99} because of its high accuracy and because it uses only a single snapshot of data as in our case. Other high accuracy time delay estimation methods, like MUSIC \cite{Li04}, are not applicable here because they require multiple uncorrelated data snapshots.
Let $\tau_\text{max}$ be the largest possible propagation delay, then in the Fuchs' technique, the continuous set of all possible propagation delays $[0,\tau_\text{max}]$ is discretized forming a grid of delays
\begin{equation} \label{eq:grid_delays}
	\mathcal{D}=
	\left\{
	0,\tau_{\text{res}},\ldots,\tau_\text{max}
	\right\},
\end{equation}
where parameter $\tau_\text{res}$ denotes the resolution of the grid.
Define the dictionary matrix stacking the received signal waveforms for all possible (discrete) delays \eqref{eq:grid_delays}:
\begin{equation} \label{eq:dictionary_Fuchs}
	\mathbf{A} =
	\begin{bmatrix}
		\mathbf{s}\left(0\right) & \cdots & \mathbf{s}\left(\tau_\text{max}\right)
	\end{bmatrix}.
\end{equation}
Then, the propagation delays of all paths from a single source reaching the $l$-th sensor are estimated by solving the following Lasso problem of the form
\begin{equation} \label{problem:Fuchs}
	\min_{\mathbf{x}} 
	\lambda\left\|\mathbf{x}\right\|_1
	+\left\|\mathbf{r}_l-\mathbf{A}\mathbf{x}\right\|_2^2,
\end{equation}
where $\lambda$ is a regularization parameter, $\|\cdot\|_1$ is the $\ell_1$-norm of a vector, and $\mathbf{r}_l$ is the received signal defined in \eqref{eq:signal_sampled}. 
Solving this convex optimization problems, results in a sparse vector $\hat{\mathbf{x}}$ whose non-zero entries indicate the estimated delays. More precisely, if the $d$-th entry of $\hat{\mathbf{x}}$ is different than zero, then a path has been detected with propagation delay $(d-1)\tau_\text{res}$. After estimating the propagation delays, Fuchs uses a maximum description length (MDL) criterion to filter out false detections. For more details on this technique see \cite{Fuchs99}, and for a better understanding on the mathematics behind the Lasso problem see \cite{Elad}.
In \cite{Fuchs99}, the time delay estimation technique was designed for real-valued signals, and assuming only a single emitting source. Here, we generalize such approach to complex valued signals by simply allowing the variables and parameters in \eqref{problem:Fuchs} to be complex. We also generalize it to multiple sources by expanding the columns of the dictionary \eqref{eq:dictionary_Fuchs} to the waveforms of all sources:
\begin{equation} \label{eq:dictionary_Fuchs_expanded}
	\mathbf{A} =
	\begin{bmatrix}
		\mathbf{s}_1\left(0\right) & \cdots & \mathbf{s}_1\left(\tau_\text{max}\right)
		& \boldsymbol{\cdots} &
		\mathbf{s}_Q\left(0\right) & \cdots & \mathbf{s}_Q\left(\tau_\text{max}\right)
	\end{bmatrix}.
\end{equation}
It is possible to use other delay estimation techniques. Obviously, the more accurate the delay estimation technique, the better performance would be expected from this NLOS interference mitigation. Contrary to indirect localization techniques, the goal here is not to precisely estimate the propagation delays of the first paths, but rather to estimate the propagation delays of all subsequent arrivals, and cancel them out.

Let $\tilde{\tau}^1_{ql},\ldots,\tilde{\tau}^{P_{ql}}_{ql}$ be the estimated propagation delays from source $q$ to sensor $l$, then their amplitudes may be estimated by solving a linear least squares fit
\begin{equation} \label{eq:amplitudes_estimates}
		\left\{\tilde{\alpha}^p_{ql}\right\} = 
		\argmin{\left\{\alpha^p_{ql}\right\}}
		\left\| 
		\mathbf{r}_l - \sum_{q=1}^{Q}\sum_{p=1}^{P_{ql}}
		\alpha^p_{ql}
		\mathbf{s}_q\left(\tilde{\tau}^p_{ql}\right)
		\right\|_2^2.
\end{equation}
Assuming the estimated propagation delays are ordered in ascending order $\tilde{\tau}^1_{ql}<\ldots<\tilde{\tau}_{ql}^{P_{ql}}$, then all arrivals, except the first, can be canceled out from the received signals
\begin{equation} \label{eq:signal_cancelled}
	\tilde{\mathbf{r}}_l =
	\mathbf{r}_l - \sum_{q=1}^{Q}\sum_{p=2}^{P_{ql}}
	\tilde{\alpha}^p_{ql}
	\mathbf{s}_q\left(\tilde{\tau}^p_{ql}\right).
\end{equation}
Ideally, all NLOS arrivals would be perfectly detected and their propagation delays estimated, in which case we could continue with a direct localization technique designed for absent multipath. However, \eqref{eq:signal_cancelled} is not guaranteed to cancel all NLOS components for two reasons. First, if the LOS path between a source and sensor is blocked, then the first arrival corresponds to a NLOS paths, in which case it is not removed. Also, it is possible that the chosen delay estimation technique misses some arrivals or detects some false ones, thus failing to remove some NLOS components or adding some extra components, respectively. In short, this stage is essential as it reduces the multipath, but does not necessarily remove it completely. In the next section, we present a localization technique designed to work in the presence of the residual multipath as well as blocked paths.

\section{Stage 2: Localization} \label{sec:2ndPhase}

This stage seeks to estimate the sources locations using the signals $\{\tilde{\mathbf{r}}_l\}_{l=1}^L$ output by Stage 1.
As explained in the preceding section, such signals include LOS and also NLOS components, therefore, the signal model introduced in \eqref{eq:signal_sampled} for $\mathbf{r}_l$ is also valid for $\tilde{\mathbf{r}}_l$. Obviously, since $\tilde{\mathbf{r}}_l$ and $\mathbf{r}_l$ are different, so are the values of the parameters appearing in \eqref{eq:signal_sampled} that characterize them.
From here on, to keep the notation in check, we abuse the notation by writing $\mathbf{r}_l$ instead of $\tilde{\mathbf{r}}_l$. However, always bear in mind that the observations in this stage are the signals output by Stage 1.

To compute the MLE \eqref{problem:ML}, it is required that the number of LOS and NLOS paths be known, otherwise the minimization \eqref{problem:ML} tends towards a nonsensical solution with a very large number of paths, a problem known as noise overfitting \cite{Koller}.
In this section, it is first assumed that the number of sensors that receive a LOS path from the $q$-th source, say $S_q$, is known. It will be shown later in Section~\ref{sub:unknown_uQ} that such information is not really needed. Nevertheless, even if $\{S_q\}_{q=1}^Q$ are known, but since the number of NLOS paths is not, a pure MLE approach is still not feasible. 
To bypass this issue, we will rely on the fact that the number of sources and NLOS paths is relatively small with respect to the number of observations.

Define a \emph{LOS atom} as the $N\times L$ matrix of measurements of LOS paths of a signal $s_q(t)$ emitted from location $\mathbf{p}$ and received at the $L$ sensors, i.e.,
\begin{equation} \label{eq:atoms_LOS}
	\mathbf{L}_q\left(\mathbf{b},\mathbf{p}\right) =
	\begin{bmatrix}
		b(1) \mathbf{s}_q\left(\tau_1(\mathbf{p})\right)
		& \cdots &
		b(L) \mathbf{s}_q\left(\tau_L(\mathbf{p})\right)
	\end{bmatrix}
\end{equation}
where $\mathbf{b}=[b(1) \cdots b(L)]^T$ are the complex amplitudes of the LOS components. It is important to normalize $\mathbf{b}$ as it will be discussed shortly. Hence, $\|\mathbf{b}\|_2$ is constrained to a given value that we will denote $u_q$, i.e., $\|\mathbf{b}\|_2=u_q$.
Define a \emph{NLOS atom} as the $N\times L$ matrix of measurements due to a single NLOS path from the $q$-th source to the $l$-th sensor
\begin{equation} 
	\mathbf{N}_{ql}\left(\tau\right) =  e^{i\phi} \mathbf{s}_q(\tau)\mathbf{v}_l^T
	\label{eq:atoms_NLOS}
\end{equation}
where the phase and delay are $\phi$ and $\tau$, respectively, and $\mathbf{v}_l$ is a unit vector, with the unit entry indexed by $l$. Note that the dependence of $\mathbf{N}_{ql}(\tau)$ with respect to $\phi$ is omitted in the notation for simplicity reasons. Let $\hat{\mathbf{R}}$ denote the matrix of received signals \eqref{eq:signal_matrix} in the absence of noise.
Then, $\hat{\mathbf{R}}$ may be expressed as a positive linear combination of given atoms
\begin{equation} \label{eq:signal_atomic_decomp_compact}
	\hat{\mathbf{R}} =
	\sum_{k}
	c^{(k)}\mathbf{A}^{(k)},
	\quad\mathbf{A}^{(k)}\in\mathcal{A}
\end{equation}
where $c^{(k)}>0$ for all $k$, and $\mathcal{A}$ is the set of all atoms (or atomic set). The atomic set includes all different LOS and NLOS atoms,
\begin{equation} \label{eq:atomic_set}
	\mathcal{A}= \mathcal{A}_{\text{LOS}} \cup \mathcal{A}_{\text{NLOS}}
\end{equation}
where $\mathcal{A}_{\text{LOS}}$
\begin{equation}
	\mathcal{A}_{\text{LOS}} =
	\bigcup_{q=1}^Q
	\Big\{
		\mathbf{L}_q\left(\mathbf{b},\mathbf{p}\right)
		: \mathbf{b}\in\mathbb{C}^{L}, \mathbf{p}\in\mathcal{S}\subset\mathbb{R}^{2}, \left\|\mathbf{b}\right\|_2=u_q
	\Big\} \label{eq:dictionary_LOS}
\end{equation}
and $\mathcal{A}_{\text{NLOS}}$
\begin{equation}
	\mathcal{A}_{\text{NLOS}} =
	\bigcup_{q=1}^Q\bigcup_{l=1}^L
	\Big\{
		\mathbf{N}_{ql}\left(\tau\right) : 0\leq\phi<2\pi, \tau\in[0,\tau_{\max}]
	\Big\}. \label{eq:dictionary_NLOS}
\end{equation}
Here, $\mathcal{S}$ denotes the search area of the sources and $\tau_{\max}$ the maximum possible delay in the system. Notice that the set of LOS atoms and the set of NLOS atoms are infinite in the sense that $\mathbf{p}$ and $\tau$ are continuous variables within their domains. Thus this framework is inherently different in that the discrete dictionaries used in traditional compressive sensing.

Since the atomic sets are infinite, determining the coefficients $c^{(k)}$ from measurements $\hat{\mathbf{R}}$ is a highly undetermined problem. This problem is solved by seeking a sparse or simple solution in some sense to the coefficients $c^{(k)}$. As motivated in \cite{Chandrasekaran12}, this can be accomplished with the help of the concept of the \emph{atomic norm}.
More precisely, the atomic norm $\|\cdot\|_{\mathcal{A}}$ induced by the set $\mathcal{A}$ is defined as
\begin{equation} \label{eq:atomic_norm}
	\left\|\hat{\mathbf{R}}\right\|_{\mathcal{A}} =
	\inf_{c^{(k)}>0}\left\{
		\sum_{k} c^{(k)} :
		\hat{\mathbf{R}}=
		\sum_{k}c^{(k)} \mathbf{A}^{(k)},\mathbf{A}^{(k)}\in\mathcal{A}
	\right\}.
\end{equation}
An \emph{atomic decomposition of $\hat{\mathbf{R}}$} is any set of coefficients $\{c^{(k)}\}$ for given atoms $\{\mathbf{A}^{(k)}\}$ such that $\hat{\mathbf{R}}=\sum_{k}c^{(k)} \mathbf{A}^{(k)}$. The \emph{cost} of an atomic decomposition is defined as the sum of its positive coefficients: $\sum_{k}c^{(k)}$. An atomic decomposition is \emph{optimal} if its cost achieves $\|\hat{\mathbf{R}}\|_{\mathcal{A}}$, or equivalently, if its cost is the smallest among all atomic decompositions.
Sparsity is imposed here in the sense that we assume that the coefficients $c^{(k)}$ for which the atomic decomposition is optimal (i.e., lowest cost) are associated with the true solution of locations and time delays. This sparsity condition resolves the undetermined nature of \eqref{eq:signal_atomic_decomp_compact}.
In practice, in the presence of noise, we seek the optimal atomic decomposition that \emph{approximately} matches the received signals. Precisely, in \cite{Chandrasekaran12} it is suggested that the noiseless signals $\hat{\mathbf{R}}$ may be estimated by minimizing
\begin{subequations} \label{problem:atomic}
	\begin{align}
		\min_{\hat{\mathbf{R}}} &\;\;
			\left\|\hat{\mathbf{R}}\right\|_{\mathcal{A}} \label{problem:atomic_objective} \\
		\text{s.t.} &\;\;
			\left\|\mathbf{R}-\hat{\mathbf{R}}\right\|^2_F\leq\epsilon, \label{problem:atomic_error}
	\end{align}
\end{subequations}
where $\|\cdot\|_F$ is the Frobenius norm, i.e., $\|\mathbf{M}\|_f = \sqrt{\sum_{i,j}|{M}(i,j)|^2}$, for any matrix $\mathbf{M}$, where ${M}(i,j)$ is the entry at the $i$-th row and $j$-th column.
Roughly speaking, minimizing the atomic norm \eqref{problem:atomic_objective} enforces sparsity, while constraint \eqref{problem:atomic_error} sets a bound on the mismatch between the noisy signals and the estimated signals. In fact, the left hand side of \eqref{problem:atomic_error} is an ML-like cost function \eqref{problem:ML}, hence, parameter $\epsilon$ may be regarded as an educated guess of the ML cost. The optimum solution to problem \eqref{problem:atomic}, say $\hat{\mathbf{R}}^\star$, may be regarded as an estimate of the received signals in the absence of noise.
However, notice that solving such problem produces $\hat{\mathbf{R}}^\star$ only and not its optimal atomic decomposition.
Thus, in general, in order to recover the optimal atomic decomposition, first, the optimum $\hat{\mathbf{R}}^\star$ to problem  \eqref{problem:atomic} is computed, and second, the optimal atomic decomposition of $\hat{\mathbf{R}}^\star$ is found.

The atomic decomposition of $\hat{\mathbf{R}}^\star$ may be expressed
\begin{equation} \label{eq:atomic_decomposition}
	\hat{\mathbf{R}}^\star = \sum_{q=1}^{Q}
	\sum_{k=1}^{K_q}
	c_{q}^{(k)}\mathbf{L}_q\left(\mathbf{b}_q^{(k)},\mathbf{p}_q^{(k)}\right)
	+
	\sum_{q=1}^{Q}\sum_{l=1}^{L}\sum_{k=1}^{K_{ql}}
	c_{ql}^{(k)}\mathbf{N}_{ql}\left(\tau_{ql}^{(k)}\right)
\end{equation}
where $\{c_{q}^{(k)}\}_{k=1}^{K_q}$ are the positive coefficients associated to the $K_q$ \emph{non-zero} LOS atoms from the $q$-th source, and $\{c_{ql}^{(k)}\}_{k=1}^{K_{ql}}$ are the positive coefficients associated to the $K_{ql}$ \emph{non-zero} NLOS atoms between the source-sensor pair $(q,l)$.
Given $\mathbf{\hat{R}}^\star$ is expressed as in \eqref{eq:atomic_decomposition} and given that \eqref{eq:atomic_decomposition} is an optimal atomic decomposition, i.e.,  its cost
\begin{equation} \label{eq:cost}
	C=\sum_{q=1}^{Q}\sum_{k=1}^{K_q} c_{q}^{(k)} +\sum_{q=1}^{Q}\sum_{l=1}^{L}\sum_{k=1}^{K_{ql}} c_{ql}^{(k)}
\end{equation}
is the smallest, then the set of locations for the $q$-th source associated with the optimal atomic decomposition is
\begin{equation} \label{eq:optimal_locations}
	\left\{\mathbf{p}_q^{(k)} \text{ for all }  k=1,\ldots,K_q\right\},
\end{equation}
the set of LOS propagation delays between source $q$ and sensor $l$ is
\begin{equation} \label{eq:optimal_LOS_delays}
	\left\{\tau_{l}\left(\mathbf{p}_{q}^{(k)}\right)
	: b_q^{(k)}(l)\neq0 \text{, for all }   k=1,\ldots,K_{q}\right\},
\end{equation}
and the set of NLOS propagation delays between source $q$ and sensor $l$ is
\begin{equation} \label{eq:optimal_NLOS_delays}
	\left\{\tau_{ql}^{(k)} \text{ for all }  k=1,\ldots,K_{ql}\right\}.
\end{equation}
Next, a definition of correct recovery is provided.

\begin{definition} \label{def:correct_recovery}
	Given $\mathbf{\hat{R}}^\star$ is expressed as in \eqref{eq:atomic_decomposition} and given that \eqref{eq:atomic_decomposition} is an optimal atomic decomposition, then the sources locations are correctly recovered if
	\begin{align}
		K_q &= 1 \\
		\mathbf{p}_q^{(1)} &= \mathbf{p}_q,
	\end{align}
	for $q=1,\ldots,Q$.
\end{definition}
Condition $K_q=1$ is required for all $q$ because, obviously, it exists only one valid location for each source, and in such case $\mathbf{p}_q^{(1)}$ must match the true location of the $q$-th source.

In Table~\ref{fig:atomic_process}, the procedure for recovering the sources' locations from the received signals is summarized.
\begin{figure}
	\centering
	\includegraphics[width=\columnwidth]{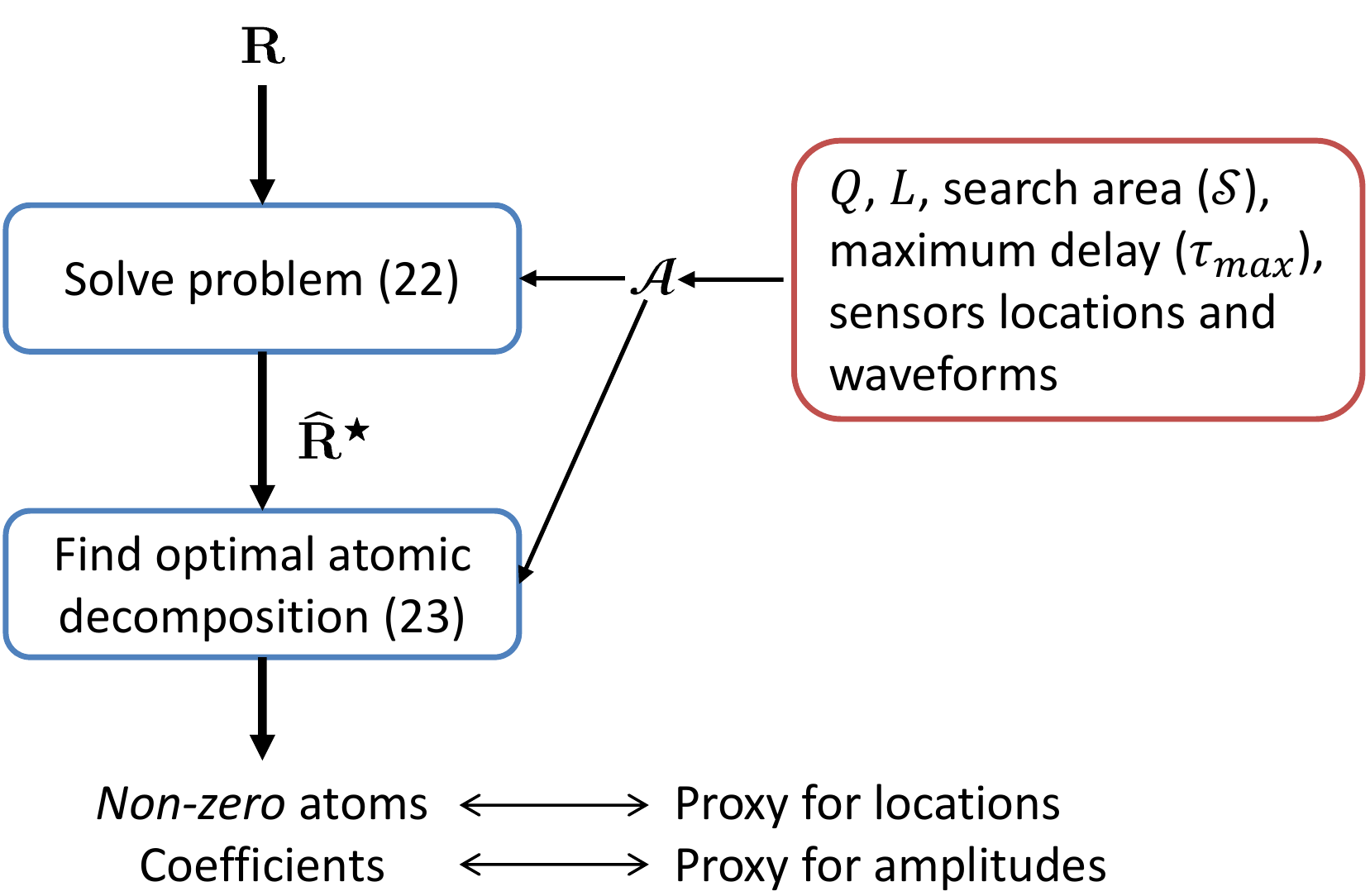}
	\caption{Flow diagram of the process for recovering the sources' locations.}
	\label{fig:atomic_process}
\end{figure}
In some specific cases, such as estimating frequencies from a mixture of complex sinusoids \cite{Tang13-2}, some sophisticated techniques have been devised for minimizing the atomic norm, and then recover the optimal atomic decomposition of $\hat{\mathbf{R}}^\star$, thanks to the particular structure of the atomic set. However, in general, it is challenging  to solve \eqref{problem:atomic}, because computing the atomic norm is not always straightforward. In Section \ref{sub:discretization}, an approximate method based on discretizing the atomic set is proposed for simultaneously solving the atomic norm minimization problem \eqref{problem:atomic} and recovering the optimal atomic decomposition \eqref{eq:atomic_decomposition}.
Before delving into the details on how to actually solve problem \eqref{problem:atomic} and find the optimal atomic decomposition as expressed in \eqref{eq:atomic_decomposition}, it is shown in the next section that tuning the parameters $\{u_q\}_{q=1}^Q$ appropriately is critical to the correct recovery of the sources' locations.

\subsection{Guarantee for Correct Recovery of the Sources' Locations} \label{sub:detection_LOS/NLOS}

In this part are developed guarantees for correct recovery of the sources' location in the sense of Definition~\ref{def:correct_recovery}. To ensure an identifiable signal model, the following assumption is made:
\begin{assumption} \label{ass:delays_identifiability}
	For each sensor, signal model \eqref{eq:signal_sampled} is identifiable in the sense that the observed data is explained by a unique set of delays, $\tau_l(\mathbf{p})$ for $q=1,\ldots,Q$, and $\tau_{ql}^{(m)}$ for $q=1,\ldots,Q$ and $m=1,\ldots,M_{ql}$. 
\end{assumption}

Further, to develop correct recovery guarantees, we assume noiseless observations, in which case the solution to \eqref{problem:atomic} is trivially $\hat{\mathbf{R}}^\star=\mathbf{R}$. However, as shown in the numerical section, the theoretical results obtained in this section are also meaningful in the presence of noise. The key properties that are exploited to obtain guarantees are:
\begin{enumerate}
	\item LOS signal paths associated with a source have a common location (see \eqref{eq:signal_LOS}).
	\item NLOS signal paths are local to sensors (see \eqref{eq:signal_NLOS}).
\end{enumerate}

To formalize the notion that LOS paths emitted by a source have a common location, we introduce the notion of location consistency:

\begin{definition} \label{def:consistency}
	A location $\mathbf{p}$ is said to be consistent with $X$ paths (LOS or NLOS), or vice-versa, if the propagation delays of such paths, say $\tau_1,\ldots,\tau_X$, satisfy
	\begin{equation} \label{eq:consistency}
		\tau_x = \tau_{l_x}\left(\mathbf{p}\right) \quad\text{for }x=1,\ldots,X,
	\end{equation}
	where $\{l_1,\ldots,l_X\}\subseteq\{1,\ldots,L\}$ are the indexes of the destination sensors of the $X$ paths, and $\tau_{l_x}(\mathbf{p})$ is the delay of the direct path between location $\mathbf{p}$ and sensor $l_x$. 
\end{definition}

In order to find the sources' locations exploiting the notion of consistency in Definition~\ref{def:consistency}, the following assumptions are made.

\begin{assumption} \label{ass:Sq_known}
	The number of LOS paths from source $q$, $S_q$, is known.
\end{assumption}

By its very nature, a source location cannot be consistent with any NLOS. Thus the location of the $q$-th source is consistent with exactly $S_q$ paths.

\begin{assumption} \label{ass:source_identifiability}
	Only the true location of the $q$-th source is consistent with $S_q$ paths emitted by the $q$-th source.
\end{assumption}


By Assumptions \ref{ass:Sq_known} and \ref{ass:source_identifiability}, given a source with a known emitted waveform and a known number $S$ of LOS paths, its true location is the only one consistent with $S$ paths.

From \eqref{eq:atomic_decomposition} and Definition~\ref{def:correct_recovery}, the solution containing the true locations of the sources is associated with the optimal atomic decomposition. However, from \eqref{eq:atomic_set} and the definition of atoms, namely, LOS atoms \eqref{eq:atoms_LOS} and NLOS atoms \eqref{eq:atoms_NLOS}, the optimal atomic decomposition is parameterized by the norm of the amplitudes in the LOS atoms $u_q$ \eqref{eq:atoms_LOS}. For given data $\hat{\mathbf{R}}^\star$, decreasing $u_q$ has to be balanced by an increase in the coefficients of the LOS atoms, thus raising their contribution to the cost $C$ \eqref{eq:cost}. Put another way, different values of $u_q$ lead to different explanations of the data $\hat{\mathbf{R}}^\star$ manifested as different optimal atomic decompositions, and thus corresponding to different solutions of the source localization problem. We seek to determine which values of parameters $u_q$ ensure that the corresponding optimal atomic decomposition results in locations that are consistent with the number of paths indicated by Assumption~\ref{ass:Sq_known}. This in turn guarantees that these are the true sources' locations. The next lemma establishes the condition on $u_q$ under which a location associated with the optimal atomic decomposition is also consistent with the number of LOS paths.

\begin{lemma} \label{lem:right_inequality}
	Given a known number of LOS paths $S_q$ of the $q$-th source, if parameter $u_q$ satisfies
	\begin{equation} \label{eq:right_inequality}
		u_q<\frac{1}{\sqrt{S_q-1}},
	\end{equation}
	then any location (for the $q$-th source) associated to the optimal atomic decomposition \eqref{eq:optimal_locations} is consistent, in the sense of Definition~\ref{def:consistency}, with $S_q$ or more paths.
\end{lemma}

For the proof of Lemma~\ref{lem:right_inequality}, see Appendix \ref{appen:right_inequality}. The interpretation of this lemma is that given a solution that produces a location with less than $S_q$ paths, and if condition \eqref{eq:right_inequality} is met, there exists another lower cost solution, implying that a solution with fewer than $S_q$ paths cannot be optimal.

The previous lemma guarantees that any location associated with the optimal atomic decomposition is consistent with $S_q$ paths.
The next lemma establishes the condition on $u_q$ that ensures that at least one location is associated with the optimal atomic decomposition.

\begin{lemma} \label{lem:left_inequality}
	Given a known number of LOS paths $S_q$ of the $q$-th source,
	if parameter $u_q$ satisfies
	\begin{equation} \label{eq:left_inequality}
		u_q > \frac{1}{\sqrt{S_q}},
	\end{equation}
	then at least one location (for the $q$-th source) is associated to the optimal atomic decomposition.
\end{lemma}

For the proof of Lemma~\ref{lem:left_inequality}, see Appendix \ref{appen:left_inequality}. The interpretation of this lemma is that given a solution that does not produce a location for the $q$-th source, and if condition \eqref{eq:left_inequality} is met, there exists another lower cost solution that produces a location for the $q$-th source.

The two lemmas lead directly to the following theorem establishing the guarantee for correct recovery of the sources' locations.

\begin{theorem} \label{thm:u_q}
	A sufficient condition for the correct recovery in the sense of Definition~\ref{def:correct_recovery} of the sources' locations is that
	\begin{equation} \label{eq:u_q}
		\frac{1}{\sqrt{S_q}} < u_q < \frac{1}{\sqrt{S_q-1}}
	\end{equation}
	for all $q$.
\end{theorem}

\begin{proof}
	If $u_q>\nicefrac{1}{\sqrt{S_q}}$ for all $q$, by Lemma~\ref{lem:left_inequality}, at least one location is associated to the optimal atomic decomposition for each source.
	By Assumption~\ref{ass:Sq_known}, the number of LOS paths $S_q$ is known for each source $q$. Therefore, if $u_q$ is chosen such that $u_q<\nicefrac{1}{\sqrt{S_q-1}}$ for all $q$, then by Lemma~\ref{lem:right_inequality}, the locations associated to the optimal atomic decomposition for the source $q$ are consistent with $S_q$ or more paths. However, according to Assumption~\ref{ass:source_identifiability}, only the location of the source is consistent with $S_q$ or more paths, thus completing the proof.
\end{proof}

The interpretation of Theorem~\ref{thm:u_q} is that the LOS atoms should be large enough to guarantee at least one LOS solution, but small enough to guarantee that the solution is the correct one.
A numerical examples illustrates Theorem~\ref{thm:u_q}. Let the search area be of size \SI{200x200}{\metre} and centered around the origin of the coordinate system. A single source is positioned at (\SI{20}{\meter}, \SI{30}{\metre}) and 5 sensors are positioned at coordinates (\SI{40}{\meter}, \SI{-40}{\metre}), (\SI{-40}{\meter}, \SI{-40}{\metre}), (\SI{-40}{\meter}, \SI{40}{\metre}), (\SI{40}{\meter}, \SI{40}{\metre}) and (\SI{0}{\meter}, \SI{0}{\metre}). All sensors receive a LOS path except for the sensor located at (\SI{40}{\meter}, \SI{-40}{\metre}). Therefore, the number of LOS paths is $S_1=4$. In addition, the sensor at the origin receives a NLOS path whose path length is \SI{91}{\meter}.
The goal is to compute the probability of correct recovery in the sense of Definition~\ref{def:correct_recovery} as a function of $u_1$ and under the conditions of Theorem~\ref{thm:u_q}, i.e., the noiseless case. 
The implementation of the procedure leading to Fig.~\ref{fig:correct_recovery} is discussed in Section~\ref{sub:discretization}.
To estimate the probability of correct recovery, the experiment is repeated 1000 times, and in each experiment the emitted waveform as well as the amplitudes of the LOS and NLOS paths are chosen randomly. The exact model for generating the waveforms, as well as other parameters is the same as the one detailed in Section~\ref{sec:simulations}. 
\begin{figure}
	\centering
	\includegraphics[width=\columnwidth]{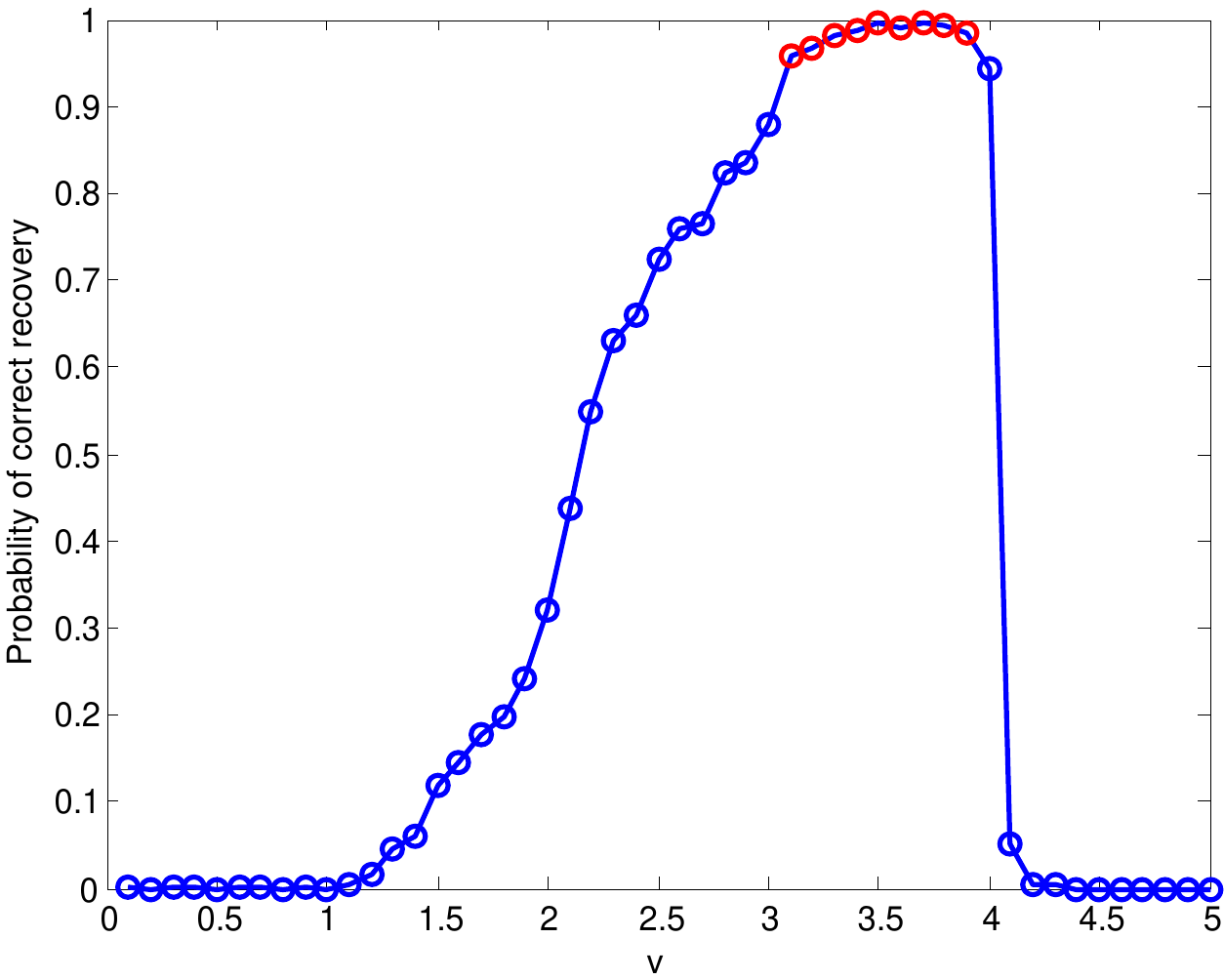}
	\caption{Probability of correct recovery in the sense of Definition~\ref{def:correct_recovery} in the absence of noise.}
	\label{fig:correct_recovery}
\end{figure}
Figure~\ref{fig:correct_recovery} plots the probability of correct recovery versus parameter $v$ which is defined as $v=(\nicefrac{1}{u_1})^2$. Theorem~\ref{thm:u_q} guarantees a correct solution if 
\begin{equation} \label{eq:v}
	S_1-1 < v < S_1.
\end{equation}
As it can be seen in Fig.~\ref{fig:correct_recovery}, for values of $v$ within the interval $]3,4[$, the probability of correct recovery is close to one, whereas it is smaller for other values.

\subsection{Practical Implementation: Discretization of the Atomic Set} \label{sub:discretization}

Remind the reader that in general, when noise is present the process for recovering the sources locations follows Fig.~\ref{fig:atomic_process}. The most straightforward method for solving problem \eqref{problem:atomic} and obtain its optimal atomic decomposition \eqref{eq:atomic_decomposition} is to substitute the atomic norm in the objective function \eqref{problem:atomic_objective} by its definition \eqref{eq:atomic_norm}, and optimize over the set of positive coefficients $\{c^{(k)}\}$.
However, such approach yields an infinite-dimensional problem because the number of atoms is infinite.  Except for some particular cases, like recovering frequencies of mixtures of sinusoids \cite{Tang13-2,Candes14}, it is in general very challenging to optimize infinite-dimensional convex problems. In \cite{Tang13}, it is advocated that dictionaries whose atoms depend on continuous parameters are discretized. For instance, the NLOS atoms \eqref{eq:atoms_NLOS} depend on a delay, which is by definition within the interval $[0,\tau_\text{max}]$, and can be discretized into a grid of discrete delays such as \eqref{eq:grid_delays}. 
In \cite{Tang13}, it is proven that the optimization problem based on the discretized atomic set converges to the original problem \eqref{problem:atomic} as the grid finesse increases. Indeed, grid refinement approaches can be found in some signal processing applications such as delay estimation \cite{Fuchs99}, direction-of-arrival estimation \cite{Malioutov05,Mahata10,Hu12} or direct localization of sources \cite{Picard10}.

The atomic set is composed of LOS \eqref{eq:atoms_LOS} and NLOS atoms \eqref{eq:atoms_NLOS}. The LOS atoms are parametrized by the location of the source whereas the NLOS atoms are parametrized by their propagation delays. Therefore, two different types of grids need to be created: one grid of locations and one grid of delays. The propagation delays of the NLOS paths vary between 0 and $\tau_\text{max}$. Upon discretizing the interval of delays with a resolution of $\tau_\text{res}$
\begin{equation} \label{eq:grid_delays_2}
	\mathcal{D}=
	\left\{
	0,\tau_{\text{res}},\ldots,\left\lfloor\frac{\tau_\text{max}}{\tau_{\text{res}}}\right\rfloor\tau_{\text{res}}
	\right\},
\end{equation}
a new set of NLOS atoms is obtained
\begin{equation}
	\tilde{\mathcal{A}}_{\text{NLOS}} =
	\bigcup_{q=1}^Q\bigcup_{l=1}^L
	\Big\{
		\mathbf{N}_{ql}\left(\tau\right) : 0\leq\phi<2\pi, \tau\in\mathcal{D}
	\Big\}.
\end{equation}
Similarly, upon discretizing the search area $\mathcal{S}$ into a uniform grid of squared cells whose center points are
\begin{equation} \label{eq:grid_locations}
	\mathcal{G}=\{\boldsymbol{\uptheta}_1,\ldots,\boldsymbol{\uptheta}_{|\mathcal{G}|}\},
\end{equation}
where the grid resolution is defined as $d_\text{res}=\min_{i\neq j}\|\boldsymbol{\uptheta}_i-\boldsymbol{\uptheta}_j\|_2$, a new set of LOS atoms is obtained
\begin{equation}
	\tilde{\mathcal{A}}_{\text{LOS}} =
	\bigcup_{q=1}^Q
	\Big\{
		\mathbf{L}_q\left(\mathbf{b},\mathbf{p}\right)
		: \mathbf{b}\in\mathbb{C}^{L}, \mathbf{p}\in\mathcal{G}\subset\mathbb{R}^{2}, \left\|\mathbf{b}\right\|_2=u_q
	\Big\} .
\end{equation}
The discrete atomic set including the LOS and NLOS atoms is
\begin{equation}
	\widetilde{\mathcal{A}} =
	\widetilde{\mathcal{A}}_{\text{LOS}} \cup \widetilde{\mathcal{A}}_{\text{NLOS}},
\end{equation}
and the atomic norm induced by $\widetilde{\mathcal{A}}$ has the same expression than in \eqref{eq:atomic_norm}
\begin{equation} \label{eq:atomic_norm_discrete}
	\left\|\hat{\mathbf{R}}\right\|_{\widetilde{\mathcal{A}}} =
	\inf_{c^{(k)}>0}\left\{
		\sum_{k} c^{(k)} :
		\hat{\mathbf{R}}=
		\sum_{k}c^{(k)} \mathbf{A}^{(k)},\mathbf{A}^{(k)}\in\tilde{\mathcal{A}}
	\right\},
\end{equation}
except for the fact that $\mathcal{A}$ has been replaced by $\tilde{\mathcal{A}}$.
By expressing the generic atoms $\mathbf{A}^{(k)}$ in \eqref{eq:atomic_norm_discrete} as LOS or NLOS atoms, the new atomic norm $\|\cdot\|_{\widetilde{\mathcal{A}}}$ may be cast as
\begin{subequations} \label{eq:new_atomic_norm_discrete}
\begin{equation} \label{eq:atomic_norm_discrete_2}
	\left\|\hat{\mathbf{R}}\right\|_{\widetilde{\mathcal{A}}} =
	\inf_{c_{q}^{(g)},c_{ql}^{(d)}\geq0}
	\Bigg\{
		\sum_{q=1}^{Q} \sum_{k=1}^{|\mathcal{G}|}
		c_{q}^{(g)}+
		\sum_{q=1}^{Q} \sum_{l=1}^{L} \sum_{g=1}^{|\mathcal{D}|}
		c_{ql}^{(d)}
\end{equation}
with coefficients $c_g^{(g)}$ and $c_{ql}^{(d)}$ such that
\begin{multline} \label{eq:new_atomic_norm_discrete_atoms}
	\hat{\mathbf{R}} = 
	\sum_{q=1}^{Q} \sum_{g=1}^{|\mathcal{G}|}
	c_{q}^{(g)} \mathbf{L}_q\left(\mathbf{b}_q^{(g)},\boldsymbol{\uptheta}_g\right)
	+\\+
	\sum_{q=1}^{Q}\sum_{l=1}^{L}\sum_{d=1}^{|\mathcal{D}|}
	c_{ql}^{(d)}\mathbf{N}_{ql}\left((d-1)\tau_\text{res}\right)
	\Bigg\},
\end{multline}
\end{subequations}
and $\mathbf{b}_q^{(g)}$ may be any vector such that $\|\mathbf{b}_q^{(g)}\|_2=u_q$ for all $q$ and $g$.
By replacing the LOS and NLOS atoms in \eqref{eq:new_atomic_norm_discrete_atoms} by their definitions \eqref{eq:atoms_LOS}-\eqref{eq:atoms_NLOS}, constraint \eqref{eq:new_atomic_norm_discrete_atoms} becomes
\begin{multline} \label{eq:new_constraint}
	\hat{\mathbf{r}}_l =
	\sum_{q=1}^{Q} \sum_{g=1}^{|\mathcal{G}|}
	c_{q}^{(g)} b_q^{(g)}(l) \mathbf{s}_q\left(\tau_l\left(\boldsymbol{\uptheta}_g\right)\right)
	+\\+
	\sum_{q=1}^{Q}\sum_{l=1}^{L}\sum_{d=1}^{|\mathcal{D}|}
	c_{ql}^{(d)}e^{i\phi_{ql}^{(d)}} \mathbf{s}_q\left((d-1)\tau_\text{res}\right)
	\\ \text{for }l=1,\ldots,L
\end{multline}
where $\hat{\mathbf{R}}=[\hat{\mathbf{r}}_1 \cdots \hat{\mathbf{r}}_L]$ and $\mathbf{b}_q^{(g)}=[b_q^{(g)}(1) \cdots b_q^{(g)}(l)]^T$.
Next, substituting the atomic norm $\|\cdot\|_{\widetilde{\mathcal{A}}}$ instead of $\|\cdot\|_{\mathcal{A}}$ in problem \eqref{problem:atomic} with \eqref{eq:atomic_norm_discrete_2} and \eqref{eq:new_constraint} yields
\begin{subequations} \label{problem:atomic_discrete}
	\begin{align}
		\min_{\substack{c_{q}^{(g)},c_{ql}^{(d)}\geq0 \\ \|\mathbf{b}_q^{(g)}\|_2 = u_q \\ 0\leq\phi_{ql}^{(d)}<2\pi}} &\;\;
			\sum_{q=1}^{Q} \sum_{k=1}^{|\mathcal{G}|}
			c_{q}^{(g)}+
			\sum_{q=1}^{Q} \sum_{l=1}^{L} \sum_{g=1}^{|\mathcal{D}|}
			c_{ql}^{(d)} \label{problem:atomic_discrete_objective} \\
		\text{s.t.}&\;\;
			\sum_{l=1}^{L}\left\|\mathbf{r}_l-\hat{\mathbf{r}}_l\right\|_2^2 \leq\epsilon \label{problem:atomic_discrete_error} \\
		&\;\;
			\begin{multlined}
				\hat{\mathbf{r}}_l =
				\sum_{q=1}^{Q} \sum_{g=1}^{|\mathcal{G}|}
				c_{q}^{(g)} b_q^{(g)}(l) \mathbf{s}_q\left(\tau_l\left(\boldsymbol{\uptheta}_g\right)\right)+
			\\
				+\sum_{q=1}^{Q}\sum_{l=1}^{L}\sum_{d=1}^{|\mathcal{D}|}
				c_{ql}^{(d)}e^{i\phi_{ql}^{(d)}} \mathbf{s}_q\left((d-1)\tau_\text{res}\right) 
			\\
				\text{for }l=1,\ldots,L.
			\end{multlined} \label{problem:atomic_discrete_estimated}
	\end{align}
\end{subequations}
Problem \eqref{problem:atomic_discrete} is not convex because of the bilinear forms, $c_{q}^{(g)} b_q^{(g)}(l)$ and $c_{ql}^{(d)}e^{i\phi_{ql}^{(d)}}$, appearing in constraint \eqref{problem:atomic_discrete_estimated}.
This can be easily remedied by the following variable changes
\begin{subequations} \label{eq:variable_change}
\begin{align}
	c_q^{(g)}\mathbf{b}_q^{(g)} &= \mathbf{y}_q^{(g)} \\
	c_{ql}^{(d)} e^{i\phi_{ql}^{(d)}} &= z_q^{(d)}(l),
\end{align}
\end{subequations}
from which it follows that
\begin{subequations} \label{eq:variable_change_normed}
\begin{gather}
	\left\|c_q^{(g)}\mathbf{b}_q^{(g)}\right\|_2 = c_q^{(g)} u_q = \left\|\mathbf{y}_q^{(g)}\right\|_2\\
	\left|c_{ql}^{(d)} e^{i\phi_{ql}^{(d)}}\right| = c_{ql}^{(d)} = \left|z_q^{(d)}(l)\right|.
\end{gather}
\end{subequations}
Combining \eqref{eq:variable_change} and \eqref{eq:variable_change_normed} with \eqref{problem:atomic_discrete_estimated} and \eqref{problem:atomic_discrete_objective}, respectively, results in the following optimization problem
\begin{subequations} \label{problem:mixed_norm}
	\begin{align}
		\min_{\substack{\mathbf{y}_q^{(g)}\\{z}_q^{(d)}(l)}} &\;\;
			\sum_{q=1}^{Q} \sum_{g=1}^{|\mathcal{G}|}
			\frac{\left\|\mathbf{y}_{q}^{(g)}\right\|_2}{u_q}+
			\sum_{q=1}^{Q} \sum_{l=1}^{L} \sum_{d=1}^{|\mathcal{D}|}
			\left|z_{q}^{(d)}(l)\right| \label{problem:mixed_norm_objective} \\
		\text{s.t.}&\;\;
			\sum_{l=1}^{L}\left\|\mathbf{r}_l-\hat{\mathbf{r}}_l\right\|_2^2 \leq\epsilon \label{problem:mixed_norm_error} \\
		&\;\;
			\begin{multlined}
				\hat{\mathbf{r}}_l = 
				\sum_{q=1}^{Q} \sum_{g=1}^{|\mathcal{G}|}
				y_{q}^{(g)}(l)\mathbf{s}_q\left(\tau_l(\boldsymbol{\uptheta}_g)\right)+
			\\
				+\sum_{q=1}^{Q}
				\sum_{d=1}^{|\mathcal{D}|}
				z_{q}^{(d)}(l) \mathbf{s}_q\left((d-1)\tau_\text{res}\right)
			\\
				\text{for }l=1,\ldots,L,
			\end{multlined}
	\end{align}
\end{subequations}
which is convex and finite-dimensional.
Problem \eqref{problem:mixed_norm} is equivalent to the latent group Lasso problem \cite{Obozinski11}, and specific algorithms for solving \eqref{problem:mixed_norm} exist in the literature \cite{Friedlander11}. Moreover, the problem also falls into the class of second-order cone programs (SOCP), a subfamily of convex problems, for which efficient algorithms are available \cite{Andersen03}. In our case, the SOCP type of algorithms resulted in the fastest computational times. 
The variable ${y}_q^{(g)}(l)$ represents the amplitude of a LOS paths from source $q$ to sensor $l$ with delay $\tau_l(\boldsymbol{\uptheta}_g)$, whereas the variable ${z}_q^{(d)}(l)$ represents the amplitude of a NLOS path from source $q$ to sensor $l$ with delay $(d-1)\tau_\text{res}$. Let $\{\hat{\mathbf{y}}_q^{(g)}\}$ and $\{\hat{z}_q^{(g)}(l)\}$ be the solutions to problem \eqref{problem:mixed_norm}. Then, the location of the $q$-th source is the grid location $\boldsymbol{\uptheta}_{g}$ for which $\|\hat{\mathbf{y}}_q^{(g)}\|_2$ is larger than zero.
Intuitively speaking, minimizing the term $\sum_{q=1}^{Q} \sum_{l=1}^{L} \sum_{d=1}^{|\mathcal{D}|} |z_{q}^{(d)}(l)|$ in 
the objective function \eqref{problem:mixed_norm_objective} induces a sparse number of NLOS paths, whereas minimizing $\sum_{q=1}^{Q} \sum_{g=1}^{|\mathcal{G}|}\|\mathbf{y}_{q}^{(g)}\|_2$ induces a sparse number of sources' locations.

\subsection{Estimation of the Number of LOS Sensors} \label{sub:unknown_uQ}

According to Theorem~\ref{thm:u_q}, we must fix $u_q$ to a value that satisfies $\nicefrac{1}{\sqrt{S_q}}<u_q<\nicefrac{1}{\sqrt{S_q-1}}$ for each source $q$, where $S_q$ is the number of sensors receiving a LOS component from the $q$-th source.
Hence, $u_q$ must be set to $u_q = \nicefrac{1}{\sqrt{S_q-\mu}}$ for a parameter $\mu\in]0,1[$.
For instance, it has been observed that a satisfactory choice is $\mu = 0.2$ as it led to the best probability of correct recovery for all experiments in Section~\ref{sec:simulations}. 
In this section, we propose a method for estimating the sources locations that not only does not require a priori knowledge on the number of LOS sensors $S_q$, but in fact estimates them. 
The method works as follows. We start by assuming that all sensors receive a LOS component from all sources, $\hat{S}_q=L$ for all $q$, and set $u_q$ such that it satisfies \eqref{eq:u_q}. Then problem \eqref{problem:mixed_norm} is solved.
According to Lemma~\ref{lem:right_inequality}, the sources' locations associated to the optimal atomic decomposition for the $q$-th source are consistent with at least $\hat{S}_q$ paths. However, by Assumption~\ref{ass:source_identifiability}, no location is consistent with more than $S_q$ paths. Therefore, if the number of LOS sensors ($\hat{S}_q> S_q$) had been overestimated, no location would be obtained for source $q$. In the next step, $\hat{S}_q$ is decreased by one for all those sources without a location estimate, and problem \eqref{problem:mixed_norm_error} is solved again. These steps are repeated until a location is obtained for each source. The last value of $\hat{S}_q$ is the estimated number LOS sensors for the $q$-th source.
This method corresponds to steps \ref{alg:0}, \ref{alg:1}, \ref{alg:3}--\ref{alg:7} of DLM's algorithm described in Section \ref{sub:algorithm}.

\subsection{Spurious Locations} \label{sub:spurious}

It is observed in numerical simulations that when the sources are off-grid ($\mathbf{p}_q\notin\mathcal{G}$ for any $q$) and/or when the propagation delays of the paths are off-grid ($\tau_{ql}^{(m)},\tau_l(\mathbf{p}_q)\notin\mathcal{D}$ for any $q$, $l$, $m$), then some spurious locations may be obtained from problem \eqref{problem:mixed_norm}. 
This phenomenon is not new and it was studied in \cite{Fuchs00} in the case of delay estimation using the $\ell_1$-norm \eqref{problem:Fuchs}. It was shown that if the propagation delay of a path is off-grid, a peak appears around such propagation delay but also secondary peaks of much weaker strength appear further apart. 

To eliminate spurious locations, it is set a simple threshold criterion. Let $\hat{\mathbf{y}}_q^{(g)}$ be, for all $q$ and $g$, the solution to problem \eqref{problem:mixed_norm}. Ideally, for each source $q$, $\|\hat{\mathbf{y}}_q^{(g)}\|$ is zero for all $g$ except if $\boldsymbol{\uptheta}_g$ matches the location of the source. However, in practice, for a given source $q$, problem \eqref{problem:mixed_norm} may produce some spurious locations, in which case $\|\hat{\mathbf{y}}_q^{(g)}\|$ may be different than zero for more than a single value of $g$. If $\boldsymbol{\uptheta}_g$ is the true location of the $q$-th source, then the $S_q$ largest components of $\hat{\mathbf{y}}_q^{(g)}$ are the amplitudes of the $S_q$ LOS paths. In contrast, if $\boldsymbol{\uptheta}_g$ is a spurious location, we have observed through numerical experimentation that some of the $S_q$ largest entries will be approximately zero. Denote $\hat{\mathbf{y}}_q^{(g)\downarrow}$ the vector with the same components than $\hat{\mathbf{y}}_q^{(g)}$, but sorted in descending order, i.e., $|\hat{y}_q^{(g)\downarrow}(1)| \geq\cdots\geq |\hat{y}_q^{(g)\downarrow}(L)|$. We propose that, for a given source $q$, all locations whose $\hat{S}_q$ strongest components do no satisfy
\begin{equation} \label{eq:threshold_simplified}
	\left\lvert\hat{y}_q^{(g)\downarrow}\left(\hat{S}_q\right)\right\rvert
	> A\,T.
\end{equation}
are dismissed.
Here, $\hat{S}_q$ is the number of LOS paths assumed for source $q$ as explained in more detail in Section~\ref{sub:unknown_uQ}, parameter $A$ is the largest signal strength  of a LOS or NLOS path
\begin{equation} \label{eq:largest_strength}
	A = \max\left(
	\max_{g,q,l} \left|\hat{y}_q^{(g)}(l)\right|,
	\max_{d,q,l} \left|\hat{z}_q^{(d)}(l)\right|
	\right),
\end{equation}
and $T$ is a value smaller than 1. For instance, in the simulations it was used $T=\nicefrac{1}{30}$, so that all locations whose signal strengths are $20\log_{10}(30)\approx \SI{30}{\decibel}$ weaker than the strongest path are discarded.
If after the threshold criterion \eqref{eq:threshold_simplified} one or more locations still remain for the $q$-th source, then the location with the largest strength is picked
\begin{equation} \label{eq:correct_location}
	\mathbf{\hat{p}}_q = \boldsymbol{\uptheta}_{\hat{g}} :
	\hat{g} = \argmax{g}\left\lvert\hat{y}_q^{(g)\downarrow}\left(\hat{S}_q\right)\right\rvert.
\end{equation}
It is important to not skip \eqref{eq:threshold_simplified}, and apply \eqref{eq:correct_location} directly.  As explained in Section~\ref{sub:unknown_uQ}, the proposed technique works by initially assuming
that the number of LOS paths for the $q$-th source is $\hat{S}_q=L$ and if no location is obtained, then successively decreasing $\hat{S}_q$ until problem \eqref{problem:mixed_norm} outputs a location. 
However, if the threshold criterion \eqref{eq:threshold_simplified} is skipped and $\hat{S}_q>S_q$, a spurious location may be erroneously selected as the correct source location instead of concluding that there is no location and that $\hat{S}_q$ needs to be decreased.

\subsection{Tuning Parameter $\epsilon$}

Parameter $\epsilon$ in optimization problem \eqref{problem:mixed_norm} constraints the fitting error between the received signals and the estimated signals. 
Such a parameter is set so that the received signals without noise are a feasible solution.
Let $\hat{\mathbf{r}}_l$ be the noiseless received signal at sensor $l$, then we require that
\begin{equation} \label{eq:error_threshold}
	\sum_{l=1}^{L}\left\|\mathbf{r}_l-\hat{\mathbf{r}}_l \right\|_2^2 
	=\sum_{l=1}^{L}\left\|\mathbf{w}_l \right\|_2^2  \leq \epsilon.
\end{equation}
If $\epsilon$ is chosen too small, then it can happen that $\sum_{l=1}^{L}\|\mathbf{w}_l\|_2^2  \nless \epsilon$, thus excluding the noiseless signals from the set of possible solutions. Because the noise $\{\mathbf{w}_l\}_{l=1}^L$ are random independent complex Gaussian vectors of length $N$, it follows that the error normalized by the noise variance $2\sigma_w^{-2}\sum_{l=1}^{L}\|\mathbf{w}_l\|_2^2$ is a Chi-square random variable with $2NL$ degrees of freedom. Thus, parameter $\epsilon$ must be set to a large enough value so that $\sum_{l=1}^{L}\|\mathbf{w}_l\|_2^2\leq\epsilon$ is satisfied with high probability, e.g.,
\begin{equation}
	\operatorname{Pr}\left(\sum_{l=1}^{L}\left\|\mathbf{w}_l \right\|_2^2\leq\epsilon\right) = \gamma.
\end{equation}
Let $\operatorname{F}(x,k)$ be the cumulative distribution function of the chi-squared distribution with $k$ degrees of freedom evaluated at $x$ and $\operatorname{F}^{-1}$ its inverse function. Then
\begin{equation} \label{eq:error_threshold_set}
	\epsilon = \frac{\sigma_w^2}{2}\operatorname{F}^{-1}\left(\gamma,2NL\right).
\end{equation}

At low signal-to-noise ratio (SNR), it is possible that the energy of the received signals is too low compared to the energy of the noise causing that $\sum_{l=1}^{L} \|\mathbf{r}_l\|_2^2 \leq \epsilon$. In such case problem \eqref{problem:mixed_norm} has the trivial solution $y_q^{(g)}(l)=z_q^{(g)}(l)=0$ for all $q$, $l$, $g$ and $d$, and it will not output any locations.  If $\sum_{l=1}^{L} \|\mathbf{r}_l\|_2^2 \leq \epsilon$, we propose to estimate the locations by finding the LOS signals that have the highest correlation with the received signals:
\begin{equation} \label{eq:recover_sources_weakSNR}
	\hat{\mathbf{p}}_q =
	\argmax{\mathbf{p}\in\mathcal{G}} \sum_{l=1}^{L}\frac{\left|\mathbf{s}_q^H\left(\tau_l\left(\mathbf{p}\right)\right)\mathbf{r}_l\right|^2}{\left\|\mathbf{s}_q^H\left(\tau_l\left(\mathbf{p}\right)\right)\right\|^2}.
\end{equation}

\subsection{Grid Refinement} \label{sub:grid_refinement}

The computational complexity of minimizing the second-order cone problem \eqref{problem:mixed_norm} is $O((Q|\mathcal{G}|+QL|\mathcal{D}|)^{3.5})$ \cite{Lobo98}. To lower it we propose a recursive grid refinement procedure inspired by the ones in \cite{Malioutov05,Mahata10,Hu12}. The optimization problem \eqref{problem:mixed_norm} employs a grid of delays in order to estimate the NLOS paths between every source-sensor pair, and a grid of locations in order to estimate the location of every source. In total, there are $Q$ grids of location and $QL$ grids of delays. In comparison to previous grid refinement approaches, ours is a more complex due to the two different types of grids used to explain the observed data. The idea behind a grid refinement procedure is to start with a coarse grid(s) and refine each grid only around the active points.
Let $\tau_\text{res}$ and $d_\text{res}$ be the grid resolutions we wish to achieve in the grids of delays and locations, respectively, and suppose that in order to lower the computational complexity, the grids are refined $R$ times. If the resolution of the grids is increased by a factor of two at every step, then the grids resolutions at each step are
\begin{align}
	\tau_{\text{res},r} &= 2^{R-r} \tau_\text{res} \quad\text{for }r=1,\ldots,R \\
	d_{\text{res},r} &= 2^{R-r} d_\text{res} \quad\text{for }r=1,\ldots,R.
\end{align}
Let $\mathcal{D}_{ql,r}$ be the grid of delays for the source-sensor pair $(q,l)$ at step $r$, and $\mathcal{G}_{q,r}$ the grid of locations for source $q$.
\begin{figure}
	\centering
	\includegraphics[width=\columnwidth]{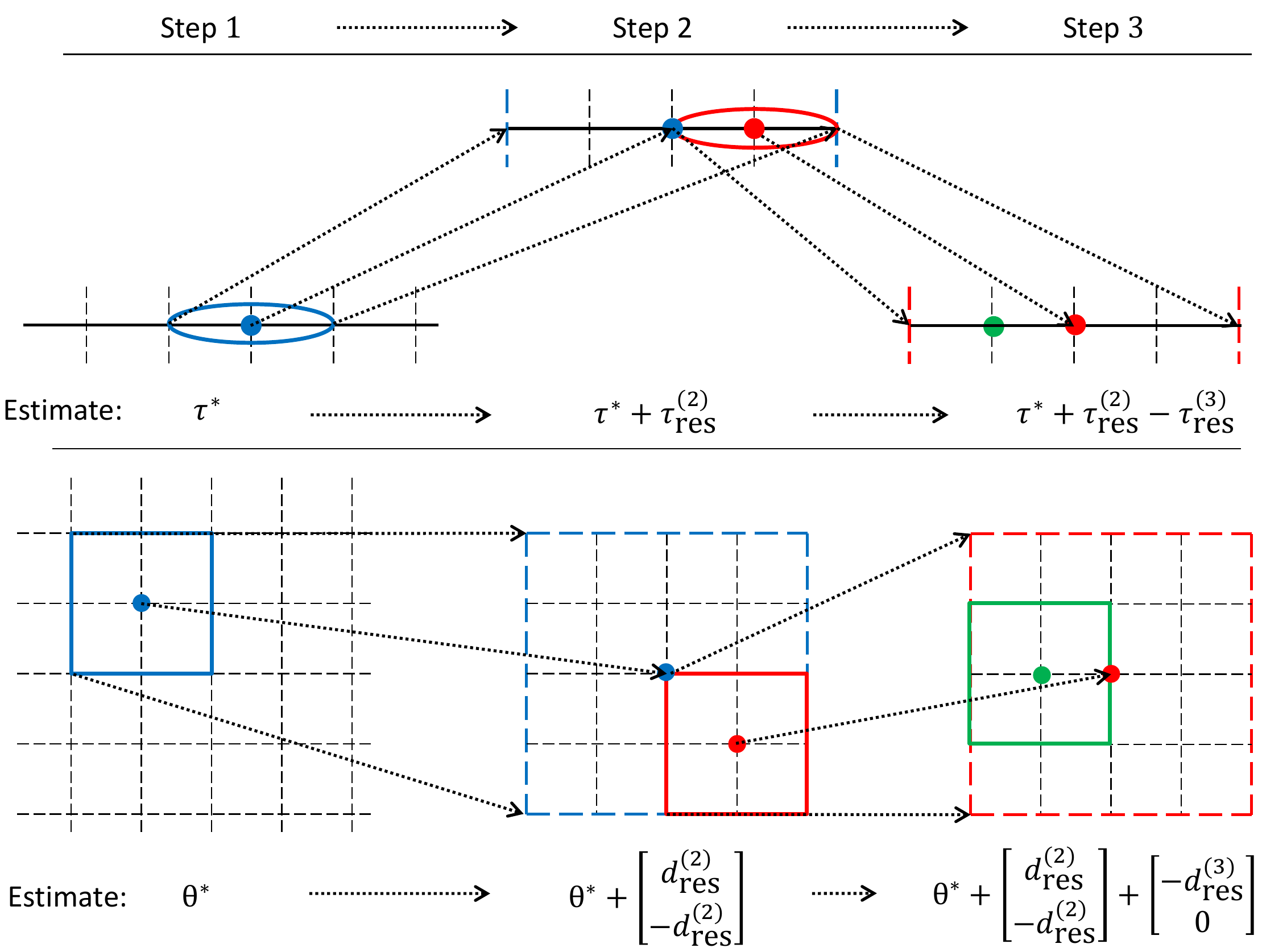}
	\caption{Illustration of three steps of a grid refinement procedure. The top image shows the grid refinement for the delays between a hypothetical source and sensor, and the bottom image shows the grid refinement for the locations of some hypothetical source, for $r=1,2,3$. The dots point out the position of a non-zero delay and location as a result of optimizing problem \eqref{problem:mixed_norm_refinement}. The positions of such non-zeros are progressively refined at each step.}
	\label{fig:grid_refinement}
\end{figure}
At the first step ($r=1$), the continuous set of delays $[0,\tau_\text{max}]$ is discretized with resolution $\tau_{\text{res},1}$
\begin{subequations} \label{eq:refinement_init}
\begin{equation}
	\mathcal{D}_{ql,1} = \left\{
		i\,\tau_{\text{res},1} \in [0,\tau_\text{max}] : i \in \mathbb{Z} 
	\right\},
	\label{eq:refinement_delays_init}
\end{equation}
and the search area $\mathcal{S}$ is discretized uniformly with resolution $d_{\text{res},1}$
\begin{equation}
	\mathcal{G}_{q,1} = \left\{
		d_{\text{res},1}
		\begin{pmatrix}
			i \\ j
		\end{pmatrix}
		\in \mathcal{S} :
		i,j \in \mathbb{Z}
	\right\}, \label{eq:refinement_locations_init}
\end{equation}
\end{subequations}
where $\mathbb{Z}$ is the set of integers.
Consider step $r$, and let the active propagation delays between the source-sensor pair $(q,l)$ be $\{\hat{\tau}_{ql,r}^{(m)} : m=1,\ldots,\hat{M}_{ql,r}\}$, and the active locations for source $q$ be $\{\hat{\mathbf{p}}_{q,r}^{(m)} : m=1,\ldots,\hat{K}_{q,r}\}$. Then, the grids at step $r+1$ include the previous active delays and locations plus some neighbor points. For instance, in addition to the active delays and locations, we include two points at the left and right of the active delays
\begin{subequations} \label{eq:refinement}
\begin{equation}
	\mathcal{D}_{ql,r+1} = \bigcup_{m=1}^{\hat{M}_{ql,r}} \left\{ \hat{\tau}_{ql,r}^{(m)}+i\,\tau_{\text{res},r+1} : i=-2,-1,0,1,2 \right\},
	\label{eq:refinement_delays}
\end{equation}
and all points within distance $2d_{\text{res},r+1}$ in the x- or y-axis of the active locations
\begin{equation}
	\mathcal{G}_{q,r+1} = \bigcup_{m=1}^{\hat{K}_{q,r}}
	\left\{
		\hat{\mathbf{p}}_{q,r}^{(m)}+
		d_{\text{res},r+1}
		\begin{pmatrix}
			i \\ j
		\end{pmatrix} :
		i,j=-2,-1,0,1,2.
	\right\}. \label{eq:refinement_locations}
\end{equation}
\end{subequations}
For a more intuitive picture on the grid refinement procedure see the examples in Figure~\ref{fig:grid_refinement} with three steps.
Next, problem \eqref{problem:mixed_norm} is solved again but only for the new grid points:
\begin{subequations} \label{problem:mixed_norm_refinement}
	\begin{align}
		\min_{\substack{\left\{\mathbf{y}_q^{(g)}\right\},\\\left\{\mathbf{z}_q^{(d)}\right\}}} &\;\;
			\sum_{q=1}^{Q} \sum_{\substack{g : \\\boldsymbol{\uptheta}_g\in\mathcal{G}_{q,r+1}}}
			\frac{\left\|\mathbf{y}_{q}^{(g)}\right\|_2}{u_q}+
			\sum_{q=1}^{Q}
			\sum_{l=1}^{L}
			\sum_{\substack{d : \\(d-1)\tau_\text{res}\\\in\mathcal{D}_{ql,r+1}}}
			\left|{z}_{q}^{(d)}(l)\right|
			\label{problem:mixed_norm_refinement_objective} \\
		\text{s.t.} &\;\;
			\sum_{l=1}^{L}\left\|\mathbf{r}_l-\hat{\mathbf{r}}_l\right\|_2^2 \leq\epsilon \label{problem:mixed_norm_refinement_error} \\
		&\;\;
			\begin{aligned}
				\hat{\mathbf{r}}_l = 
				\sum_{q=1}^{Q} \sum_{\substack{g : \\\boldsymbol{\uptheta}_g\in\mathcal{G}_{q,r+1}}}
				y_{q}^{(g)}(l)\mathbf{s}_q\left(\tau_l(\boldsymbol{\uptheta}_g)\right)+&
			\\
				+\sum_{q=1}^{Q}
				\sum_{\substack{d : \\(d-1)\tau_\text{res}\\\in\mathcal{D}_{ql,r+1}}}
				z_{q}^{(d)}(l) \mathbf{s}_q\left((d-1)\tau_\text{res}\right)&
			\\
				\text{for }l=1,\ldots,L.
			\end{aligned}
	\end{align}
\end{subequations}
The process of refining the grids and solving problem \eqref{problem:mixed_norm_refinement} is repeated for the $R$ steps. The proposed grid refinement procedure corresponds to steps \ref{alg:1.5}--\ref{alg:2} in DLM's algorithm described in Section \ref{sub:algorithm}.

In regards to the resolutions of the grids, instead of choosing the resolution of both types of grids completely independently, they are set according to
\begin{equation} \label{eq:resolution_link}
	c\,\tau_{\text{res},r} = d_{\text{res},r}\quad\text{ for any }r,
\end{equation}
where $c$ is the speed of light.

\section{Algorithm} \label{sub:algorithm}

In this section, it is presented the proposed DLM algorithm for source localization in multipath.
The inputs to the DLM algorithm are the received signals $\{\mathbf{r}_l\}_{l=1}^L$ and the noise variance $\sigma_w^2$. The number of sensors $L$, sources $Q$ and samples per sensor $N$ are assumed known. The outputs of the algorithm are the source locations estimates $\{\hat{\mathbf{p}}_q\}_{q=1}^Q$. The summary of the proposed algorithm for direct localization of RF sources in the presence of multipath is as follows:

\begin{algorithmic}[1]

	\Require $L$, $Q$, $N$, $\{\mathbf{r}_l\}_{l=1}^L$ and $\sigma_w^2$.
		
	\Param $\mathcal{S}$, $\tau_\text{max}$, $d_\text{res}$, $T$, $\gamma$.
	
	\Ensure The source locations estimates $\{\hat{\mathbf{p}}_q\}_{q=1}^Q$
	

	\Proc
		
		\For{sensor $l$ where $l=1,\ldots,L$}
	
			\State Estimate multipath TOA's $\{\tilde{\tau}_{ql}^p\}$ using \cite{Fuchs99} or any other delay estimation technique of choice.
			
			\State Estimate multipath amplitudes $\{\tilde{\alpha}_{ql}^p\}$ through \eqref{eq:amplitudes_estimates}.
			
			\State Reduce NLOS interference on the received signal $\mathbf{r}_l$ through \eqref{eq:signal_cancelled}.
			
		\EndFor
		
		\State Compute parameter $\epsilon$  through \eqref{eq:error_threshold_set}.
		
		\If{$\sum_{l-1}^{L}\|{\mathbf{r}}_l\|_2>\epsilon$}
		
			\State Compute the initial coarse grids with \eqref{eq:refinement_init} and \eqref{eq:resolution_link}.
			
			\State Initialize $\hat{S}_q = L$ for $q=1,\ldots,Q$ \label{alg:-1}
			
			\While{$\hat{\mathbf{p}}_q=\emptyset$ for any $q\in\{1,\ldots,Q\}$} \label{alg:0}
			
				\State $u_q = \frac{1}{\sqrt{\hat{S}_q-0.2}}$ for $q=1,\ldots,Q$ \label{alg:1}
				
				\For{$r=1,\ldots,R$} \label{alg:1.5}
					
					\State Optimize problem \eqref{problem:mixed_norm_refinement}. Output: $\{\hat{\mathbf{y}}_{q,r}^{(g)}\}$ and $\{\hat{{z}}_{q,r}^{(d)}(l)\}$. \label{alg:conic}
					
					\State Find the active locations $\{\hat{\mathbf{p}}_{q,r}^{(m)}\}$ and the active delays $\{\hat{\tau}_{ql,r}^{(m)}\}$.
					
					\If{$r\neq1$}
					
						\State Refine the grid with \eqref{eq:refinement} and \eqref{eq:resolution_link}.
					
					\EndIf
					
				\EndFor \label{alg:2}
				
				\State Compute $A$ through \eqref{eq:largest_strength}.
				
				\For{$q=1,\ldots,Q$} \label{alg:3}
				
				
					\If{any locations are active for the $q$-th source and such locations satisfy \eqref{eq:threshold_simplified}} \label{alg:4}
						\State Estimate the location of the $q$-th source through \eqref{eq:correct_location}.
					\ElsIf{$\hat{S}_q>1$} \label{alg:5}
						\State $\hat{S}_q \gets \hat{S}_q-1$\label{alg:6}
					\Else
						\State Estimate the location of the $q$-th source through \eqref{eq:recover_sources_weakSNR}.
					\EndIf
				\EndFor \label{alg:7}
			\EndWhile
			
		\Else 
		
			\State Recover sources' locations through \eqref{eq:recover_sources_weakSNR}.
			
		\EndIf
\end{algorithmic}

\section{Numerical Results} \label{sec:simulations}

In this section, we illustrate the performance of the localization method by numerical examples, and compare it to other existing techniques via Monte Carlo simulations. In all examples, the sources and sensors are positioned within a square area of \SI{200x200}{\metre}, which is divided into a grid of \SI{1x1}{\metre} cells, thus resulting in 40,000 cells. Unless stated otherwise, we simulate a scenario containing one source positioned at coordinates (\SI{20}{\meter},\SI{30}{\metre}) and 5 sensors positioned at coordinates (\SI{40}{\meter}, \SI{-55}{\metre}), (\SI{-45}{\meter}, \SI{-40}{\metre}), (\SI{-50}{\meter}, \SI{55}{\metre}), (\SI{60}{\meter}, \SI{60}{\metre}) and (\SI{5}{\meter}, \SI{0}{\metre}) as pictured in Fig.~\ref{fig:scenario}.
\begin{figure}
	\centering
	\includegraphics[width=\columnwidth]{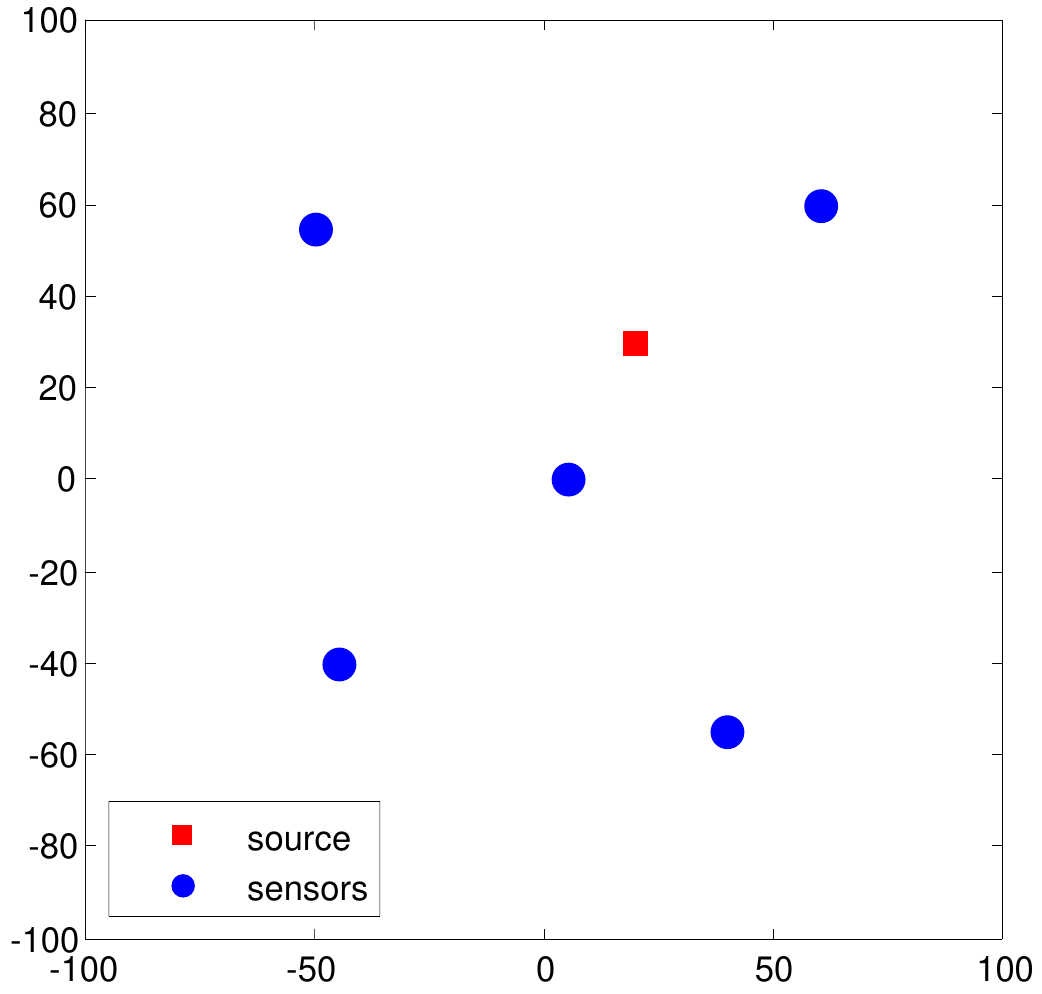}
	\caption{Map with the locations of the sensors and source used in many of the experiments in Section \ref{sec:simulations}.}
	\label{fig:scenario}
\end{figure}
The signals emitted by the sources are drawn from a white Gaussian process and filtered so that their passband bandwidth is \SI{10}{\mega\hertz}. 
If multiple sources, such as in the experiment of Section~\ref{sub:multiple_sources}, the waveforms are generated independently, thus the cross-correlation between signals from different sources is low but not necessarily zero.
All sensors are time-synchronized and sample the received signals at a \SI{20}{\mega\hertz} frequency for a total time of \SI{5}{\micro\second}, thus each sensor observes 100 samples. For each source, we define the SNR per observation time as
\begin{equation}
	\text{SNR} = 10\log_{10}\left(\frac{NLP_\text{LOS}}{\sigma_w^2}\right),
\end{equation}
where $N$ is the number of observations per sensor, $L$ is the number of sensors, $P_\text{LOS}$ is the power of a LOS component, and $\sigma_w^2$ is the variance of the sampled noise. According to \cite{Turin72}, in urban and suburban areas, the signal strengths of LOS and NLOS paths may be modeled as random variables with log-normal distribution. It follows that the channel tap powers expressed in \si{\dB} are random variables with normal distribution. For our simulations, we set the standard deviation of the tap powers to \SI{10}{\dB}.
All multipath experiments simulate Turin's urban channel model \cite{Turin72}. The arrival times of NLOS components at all sensors are modelled by a Poisson process. The mean inter-arrival time is set to \SI{0.2}{\micro\second}, and the average power $\bar{P}$ of a NLOS arrival at sensor $l$ is governed by the power delay profile (PDP)
\begin{equation}
	\bar{P}_l(t) = \exp\left(-\frac{t-t_l^{(0)}}{t_{rms}}\right)
\end{equation}
where $t$ is the arrival time of the NLOS component, $t_l^{(0)}$ is the arrival time of the LOS path and $t_{rms}$ is the root mean square (rms) delay spread. An exponential PDP assigns smaller power to later arrivals. Unless otherwise stated, all LOS paths have normalized unit power. In multipath environments, it is possible that some sensors have their LOS blocked, thus at each Monte Carlo repetition one randomly selected sensor among the five receives no LOS component.

The figures compare the performance of the following two direct localization techniques:
\begin{enumerate}
	\item \textbf{DLM} --- The proposed technique.
	\item \textbf{DPD} --- Direct Position Determination as originally propose in \cite{Weiss05} for AWGN channels.
	\item \textbf{DPD with NLOS mitigation} --- In this variation, DPD is preceded by the NLOS mitigation method introduced in Section \ref{sec:1stPhase}. The goal is to show that DLM outperforms this variation of DPD, to demonstrate that DLM's high accuracy is not due only to such NLOS interference mitigation method.
	\item \textbf{Indirect, CS TOA} --- Indirect localization comprises a two-step process. In a first step, the TOA of the first path at each sensor is estimated by a delay estimation method based on compressive sensing (CS) \cite{Fuchs99}; in a second step, multilateration is performed using the well-known method developed by Chen \cite{Chen99} to mitigate the problem of potential LOS blockage on sensors.
	\item \textbf{Indirect, matched filter TOA} --- Same as previous indirect technique, except that TOA's are estimated by a threshold-based matched filter.
\end{enumerate}
To solve the conic problem in DLM (step \ref{alg:conic} of DLM's algorithm described in Section~\ref{alg:conic}) and in CS TOA, we utilize the Mosek solver \cite{Mosek}.
The bandwidth of the emitted signals limits the localization accuracy, and it is known that the ranging resolution is approximately
\begin{equation} \label{eq:signal_resolution}
	r=\frac{c}{B}
\end{equation}
where $c$ is the speed of light and $B$ is the signal bandwidth. For the particular case of a \SI{10}{\mega\hertz} bandwidth, the waveform ranging resolution is then \SI{30}{\meter}. Also, we define the probability of correct recovery for the case of a single source as
\begin{equation} \label{eq:probability_detection}
	P_c = \frac{1}{Z}\sum_{z=1}^{Z}\mathds{1}\left(|\mathbf{p}-\hat{\mathbf{p}}^{(z)}|< \zeta\right),
\end{equation}
where $\mathbf{p}$ is the true source's location, $Z$ is the number of times that the experiment is repeated, $\hat{\mathbf{p}}^{(z)}$ is the source's location estimate for the $z$-th repetition, and $\mathds{1}(\cdot)$ is the indicator function. Unless otherwise stated, the error is set to $\zeta=\nicefrac{r}{3}$, which is a value smaller than the ranging resolution $r$. In some of the tests, it is plotted the normalized root mean square error
\begin{equation}
	\text{rMSE} = \frac{1}{r}\sqrt{\frac{1}{Z}\sum_{z=1}^{Z}\left(\mathbf{p}-\hat{\mathbf{p}}^{(z)}\right)^2}.
\end{equation}
All experiments are repeated 1000 times, i.e., $Z=1000$.

\subsection{Performance in the Absence of Multipath}

This experiment's purpose is to validate that DLM performs optimally in the absence of multipath, i.e., its accuracy matches that of the DPD, which was shown to be optimal (see \cite{Weiss05}). All five sensors receive LOS components, and Turin's channel model does not apply here, since there are no NLOS paths.
\begin{figure}
	\centering
	\includegraphics[width=\columnwidth]{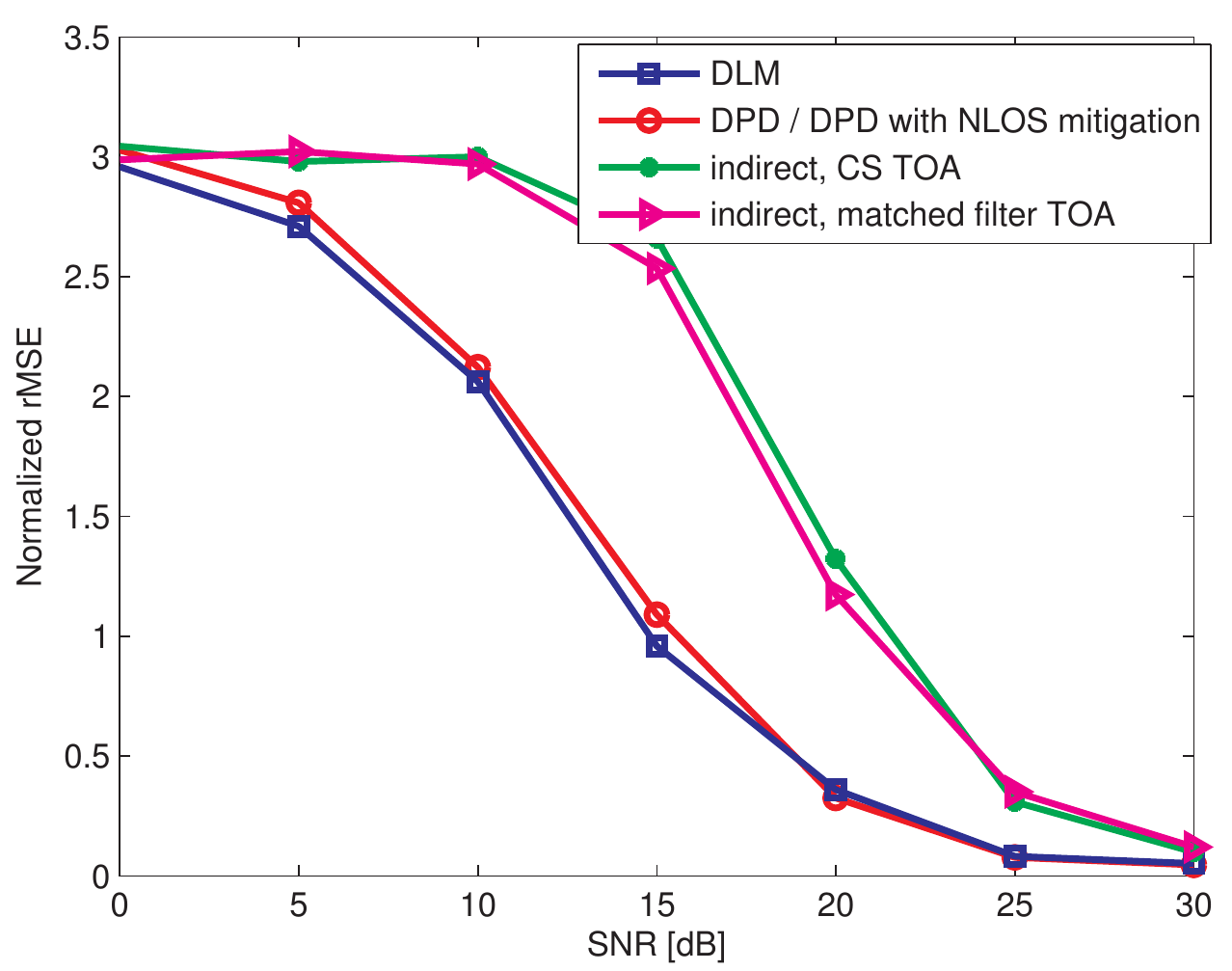}
	\caption{Root mean square error vs.\ SNR for the scenario in Fig.~\ref{fig:scenario} when no multipath is present.}
	\label{fig:noMultipath}
\end{figure}
\begin{figure}
	\centering
	\includegraphics[width=\columnwidth]{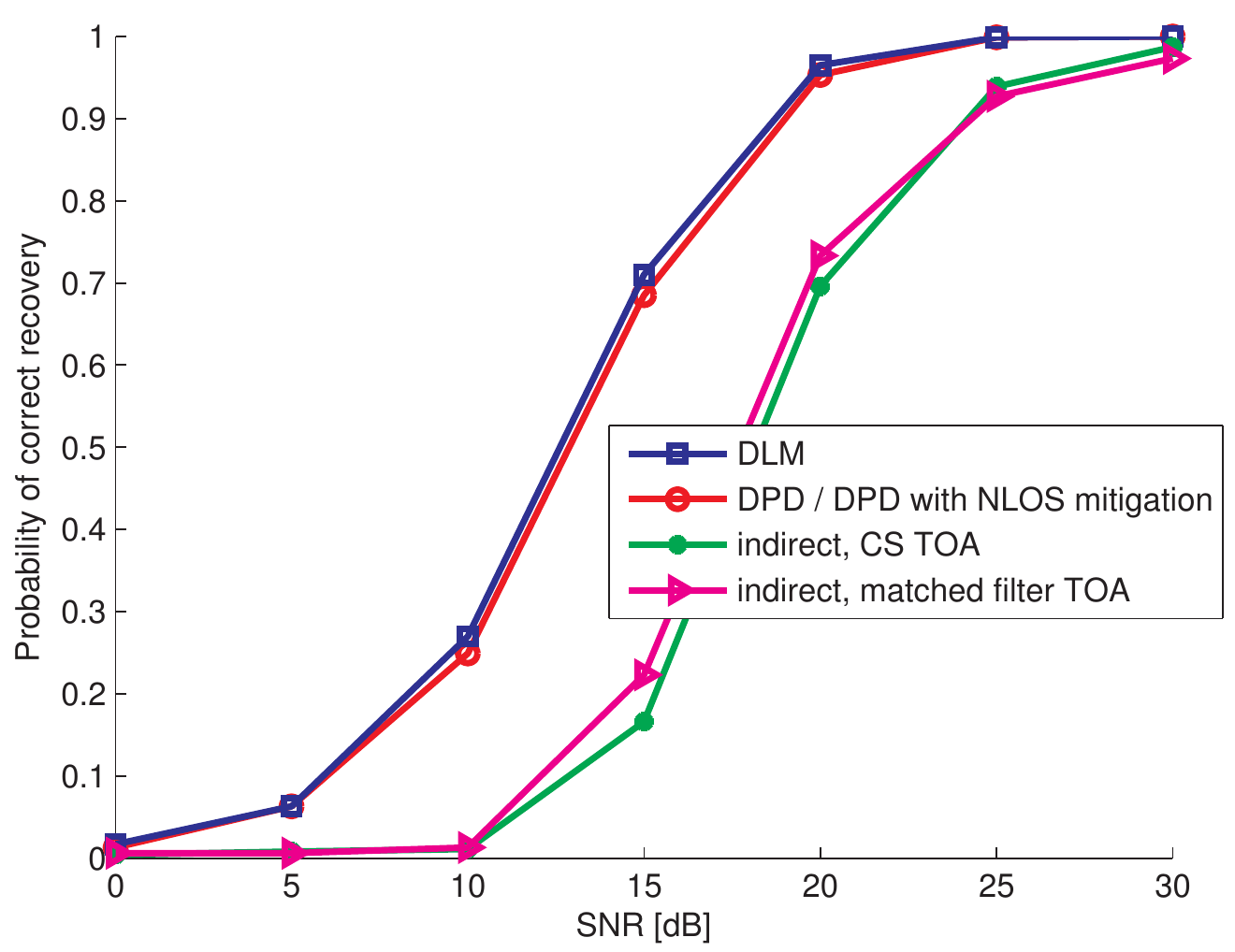}
	\caption{Probability of correct recovery vs.\ SNR for the scenario in Fig.~\ref{fig:scenario} when no multipath is present.}
	\label{fig:correctNoMultipath}
\end{figure}
Figures \ref{fig:noMultipath} and \ref{fig:correctNoMultipath} plot the rMSE and the probability of correct recovery, respectively. DPD and DPD with NLOS mitigation are plotted together because their performance is exactly the same in the absence of multipath.
As it can be observed, DPD and DLM perfom equally in terms of rMSE and probability of recovery because essentially both techniques, in the absence of multipath, look up for the location whose LOS signals correlate the most with the received signals. DPD and DLM perform substantially better in comparison to indirect techniques as it is expected from the theory.

\subsection{Performance in Multipath} \label{sub:sim_multipath}

In this example is simulated the multipath channel model described at the top of this section. The rMSE and the probability of correct recovery vs.\ SNR are plotted in Figs.\ \ref{fig:rMSE} and \ref{fig:correct}, respectively.
\begin{figure}
	\centering
	\includegraphics[width=\columnwidth]{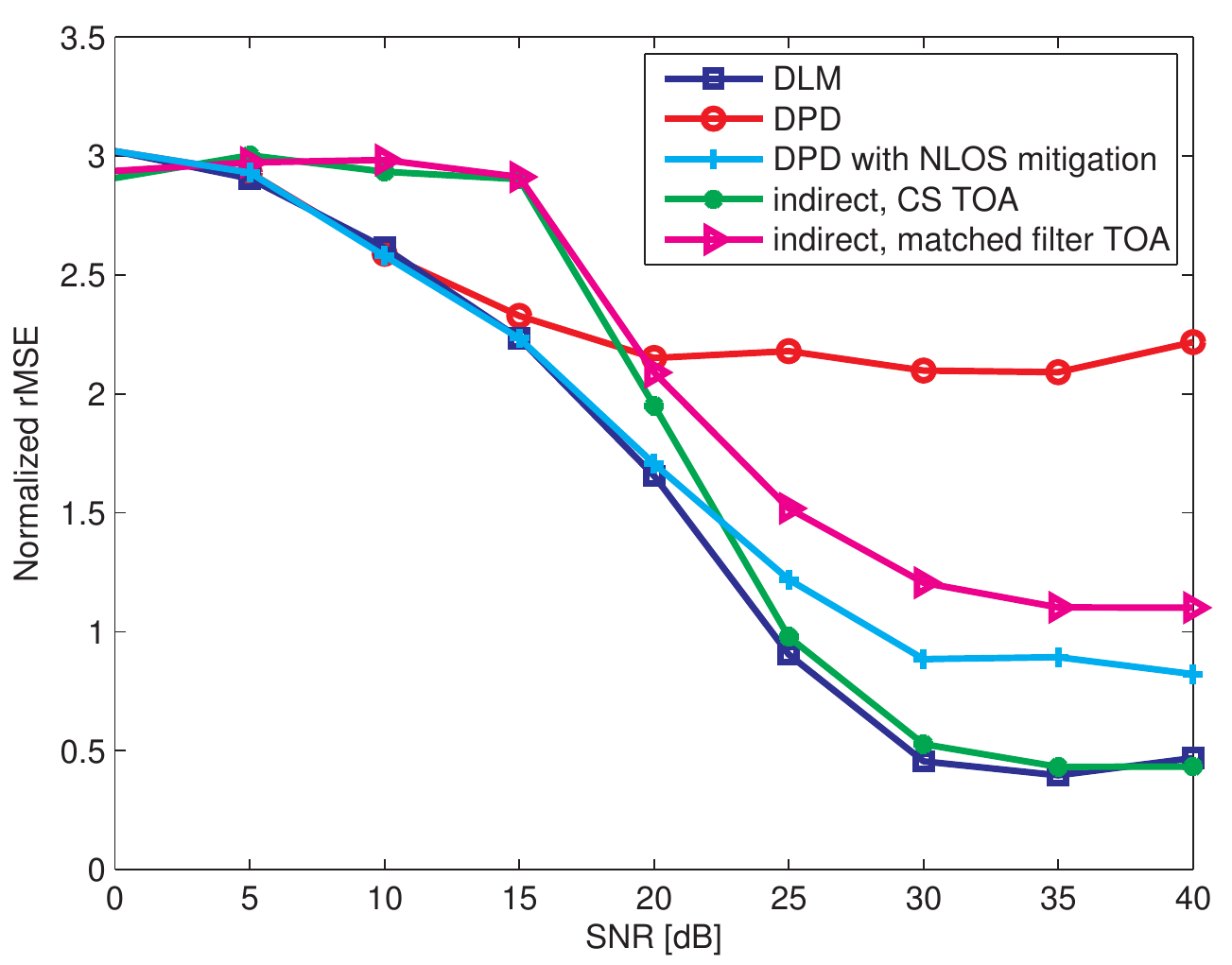}
	\caption{Root mean square error vs.\ SNR for the scenario in Fig.~\ref{fig:scenario} in a multipath environment.}
	\label{fig:rMSE}
\end{figure}
\begin{figure}
	\centering
	\includegraphics[width=\columnwidth]{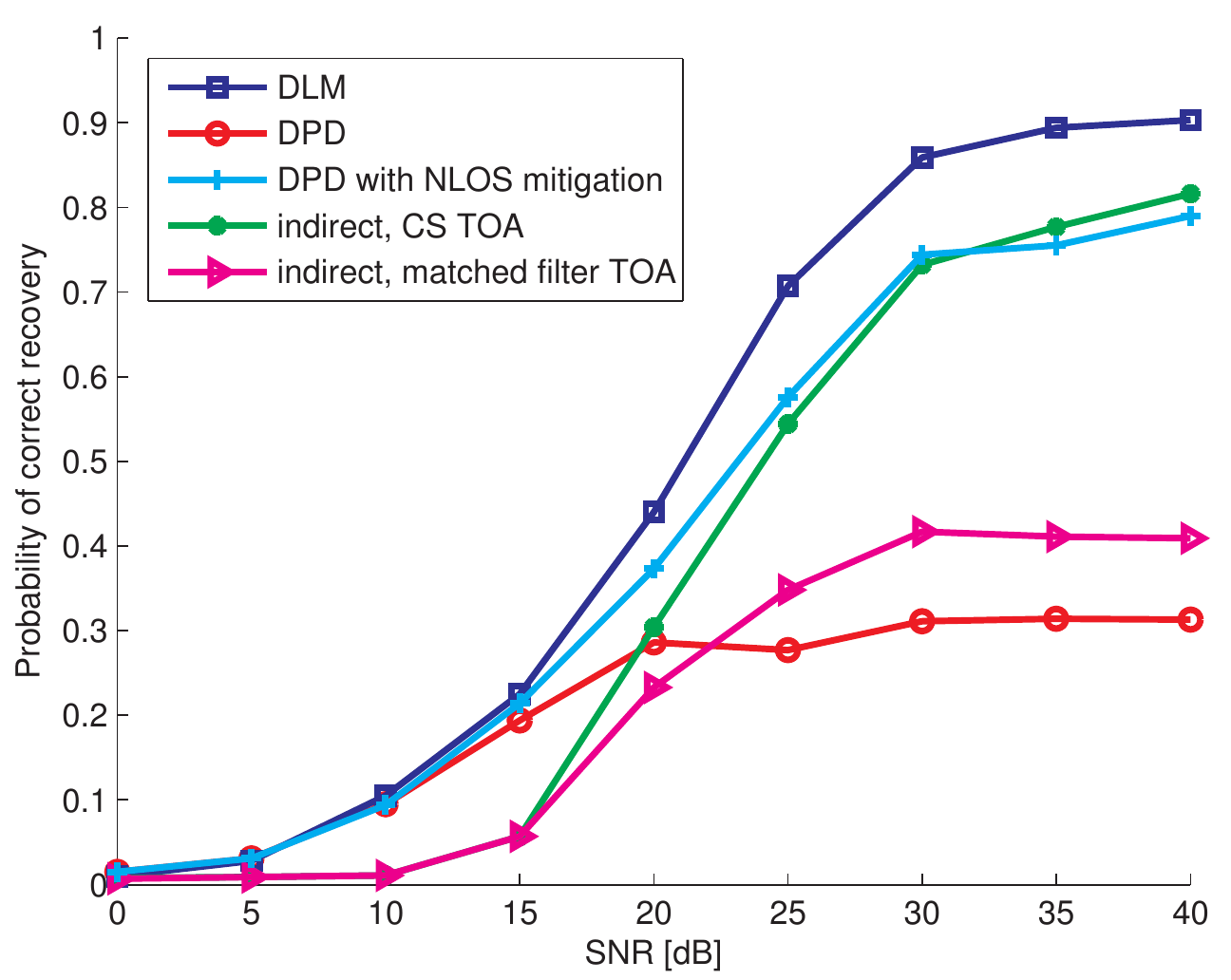}
	\caption{Probability of correct recovery vs.\ SNR for the scenario in Fig.~\ref{fig:scenario} in a multipath environment.}
	\label{fig:correct}
\end{figure}
Observe in Figs.\ \ref{fig:rMSE} and \ref{fig:correct} that DPD fails to localize the sources irrespective of the SNR due to the fact that it is not designed for multipath. Also, the indirect technique relying on estimating by matched filter the TOA of the first arrival, does not perform much better than DPD because matched filter suffers from severe bias when multiple arrivals overlap in time.
Interestingly, it seems as if DLM does not perform better, in terms of rMSE, than the indirect technique employing CS TOA estimates.
\begin{figure}
	\centering
	\includegraphics[width=\columnwidth]{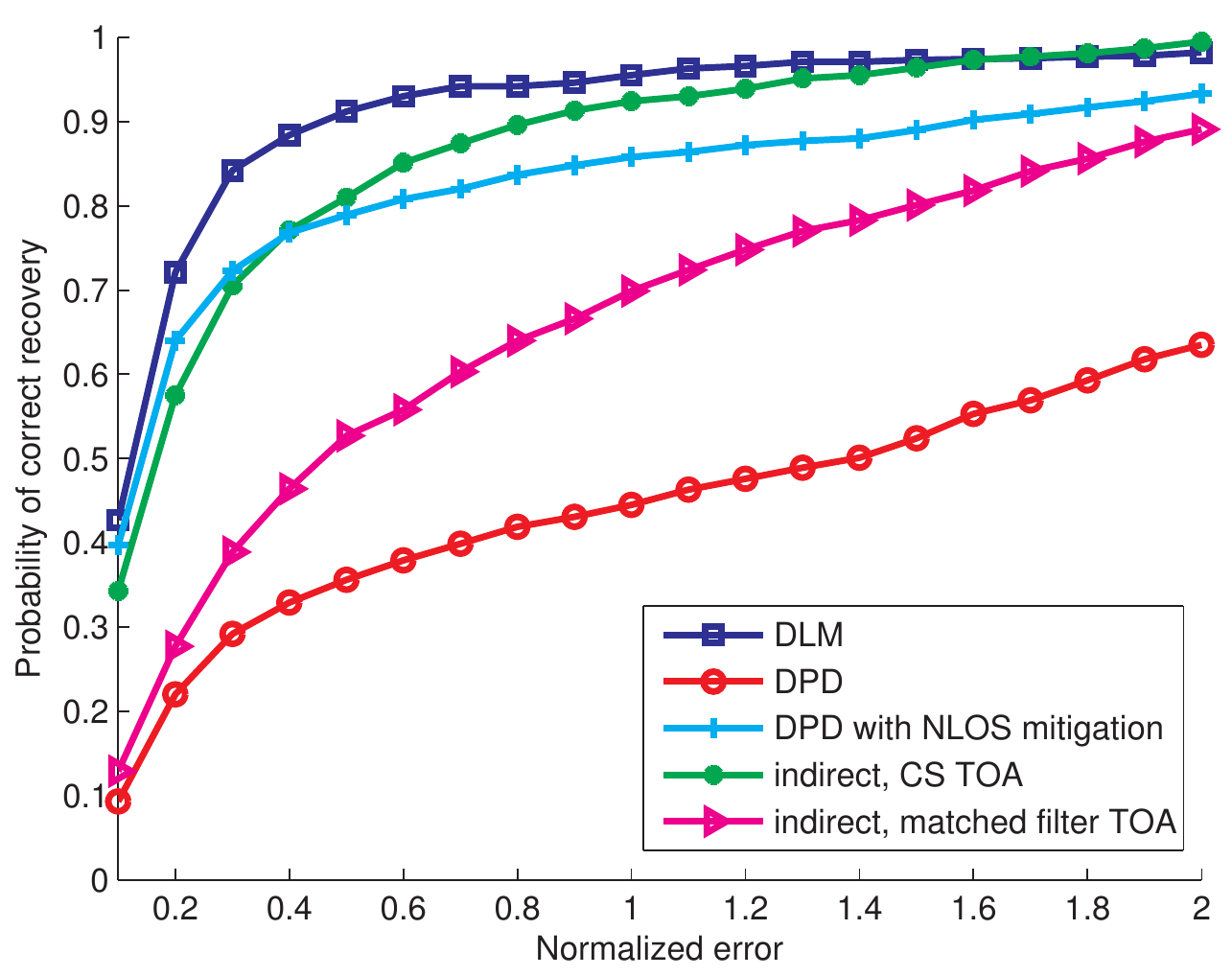}
	\caption{Probability of correct recovery vs.\ error for the scenario in Fig.~\ref{fig:scenario} for a \SI{30}{\decibel} SNR.}
	\label{fig:error}
\end{figure}
In Fig.~\ref{fig:error}, the probability of correct recovery \eqref{eq:probability_detection} is plotted for different errors ranging from 0 to $2r$ for an SNR value of \SI{30}{\decibel}. DLM achieves a high probability of correct recovery for much smaller errors than the other methods. For instance, DLM's  probability of correct recovery is 0.9 for an error smaller than $0.4r$, whereas for the indirect technique with CS TOA, such probability is only achieved when the error is $0.9r$. The other techniques perform substantially worse than DLM, and in fact, they never achieve a probability of recovery close to one even when very large errors are allowed. In summary, DLM can achieve a high probability of recovery for very small errors. In terms of rMSE, DLM and the indirect technique employing CS TOA estimates perform similarly, because in the rMSE metric small errors have a much smaller impact compared to the large errors. Hence, in the next experiments, we focus only on the probability of correct recovery.

\subsection{Probability of Correct Recovery vs.\ Delay Spread}

The considered channel model depends on the rms delay spread, which determines the interval between the LOS component and the last arriving NLOS component. In general, larger delay spreads imply more multipath that make the localization more challenging.
\begin{figure}
	\centering
	\includegraphics[width=\columnwidth]{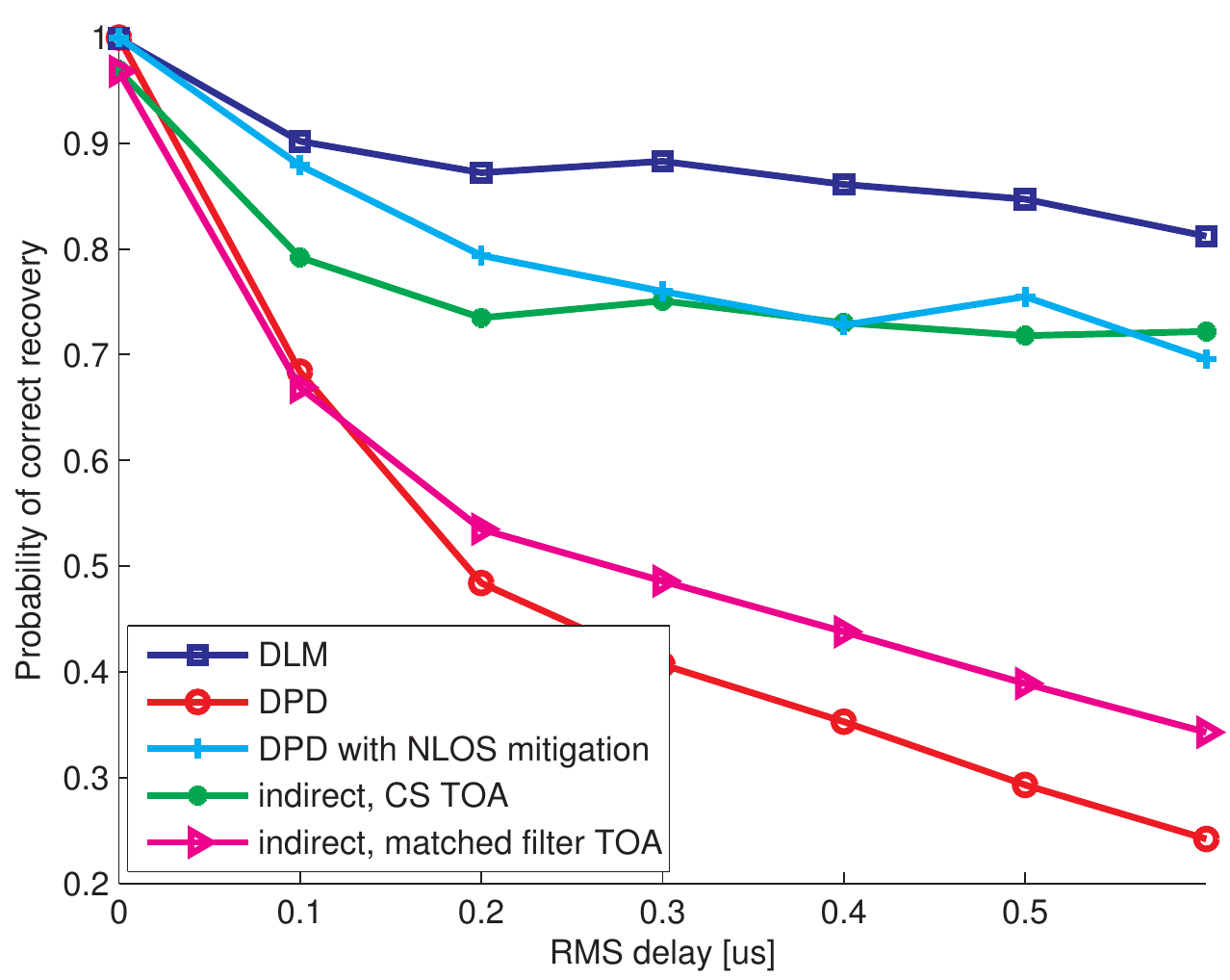}
	\caption{Probability of correct recovery vs.\ rms delay spread for the scenario in Fig.~\ref{fig:scenario} in a multipath environment for a \SI{30}{\decibel} SNR.}
	\label{fig:delay_spread}
\end{figure}
In Fig.~\ref{fig:delay_spread}, the probability of correct recovery is plotted for an rms delay spread ranging from 0 to \SI{0.6}{\micro\second} at \SI{30}{\decibel} SNR. At high-SNR and at a zero delay spread all localization techniques perform similarly. However, as soon as the rms delay spread increases by a little as \SI{0.2}{\micro\second}, DPD's performance drops markedly. The techniques specifically designed for multipath channels, such as the indirect technique based on CS TOA estimates and DLM, degrade very slightly as the rms delay spread increases. DLM outperforms all other techniques and is capable of recovering the sources locations with a high probability of correct recovery irrespective of the delay spread.

\subsection{Probability of Correct Recovery vs.\ Number of Grid Refinement Steps}

The purpose of the grid refinement procedure introduced in Section \ref{sub:grid_refinement} is to reduce the computational complexity of DLM, while maintaining the localization accuracy.
Figure~\ref{fig:grid_performance} plots the probability of correct recovery (square marker) and the DLM's mean elapsed time at Stage 2 (circle marker), versus the number of grid refinement steps.
\begin{figure}
	\centering
	\includegraphics[width=\columnwidth]{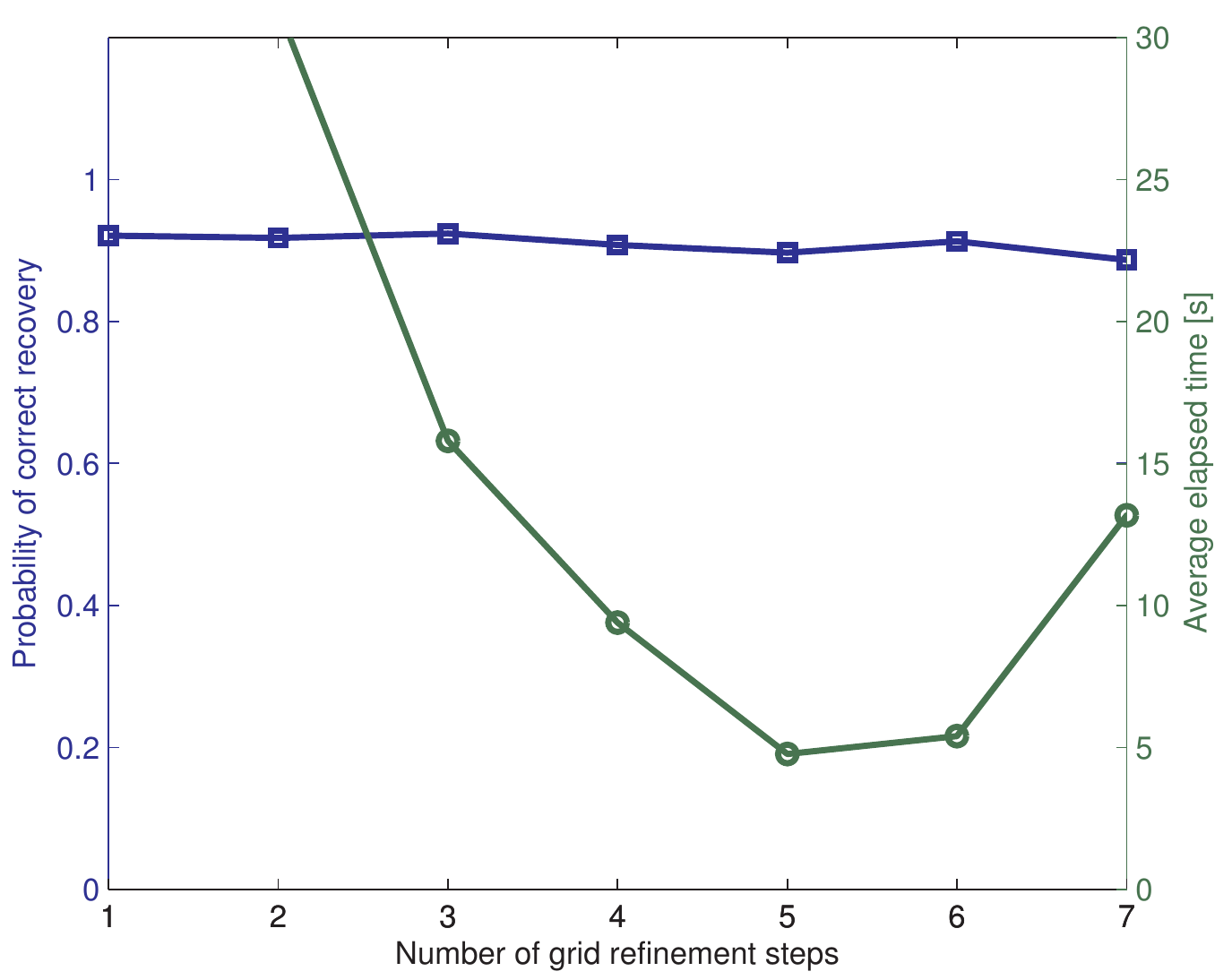}
	\caption{The left axis plots the probability of correct recovery and the right axis the mean elapsed time for running DLM's Stage 2, vs.\ the number of grid refinement steps. The SNR is fixed at \SI{30}{\decibel}.}
	\label{fig:grid_performance}
\end{figure}
The SNR is fixed at \SI{30}{\decibel}. DLM is run on a computer with an Intel Xeon processor at \SI{2.8}{\giga\hertz} with 4 GB of RAM memory. Perhaps surprisingly, the probability of correct recovery remains almost constant irrespective of the number of steps. The lowest computational time is \SI{5}{\second} and is obtained for five grid refinement steps. The number of grid steps that results in the lowest computational time depends on many factors such as number of grid points, efficiency of the conic solver, particular scenario and so forth. Thus, in general, the optimum number of steps must be found by in situ testing.

\subsection{Number of LOS sensors} \label{sub:LOSsensors}

The information about the sources' locations is carried on the LOS components (see signal model \eqref{eq:signal_sampled}). This experiment evaluates DLM's probability of correct recovery versus the number of sensors receiving a LOS path. Since the setup of Fig.~\ref{fig:scenario} includes five sensors, the number of LOS sensors is varied between one and five.
\begin{figure}
	\centering
	\includegraphics[width=\columnwidth]{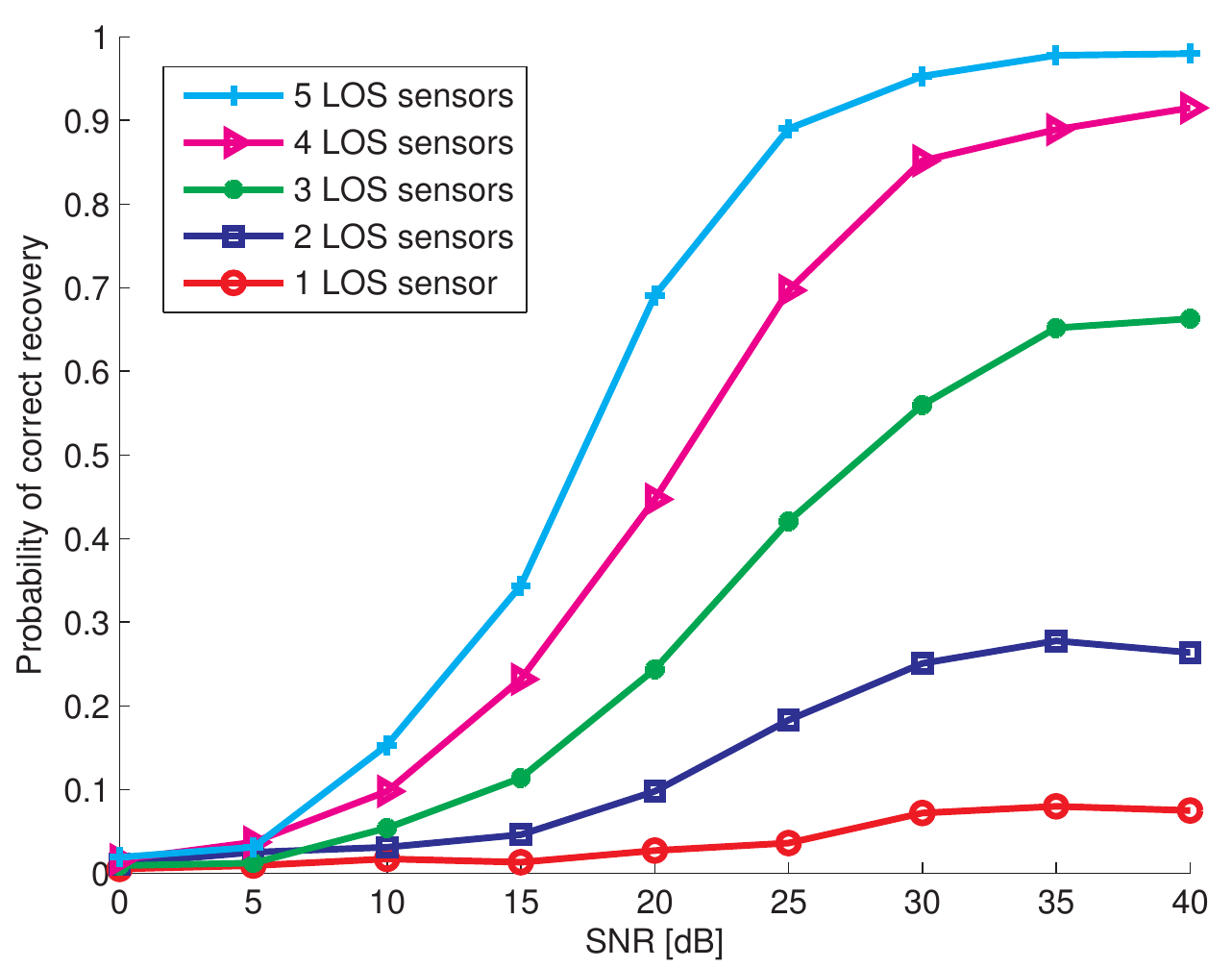}
	\caption{Probability of correct recovery vs.\ the number of LOS sensors for a \SI{30}{\decibel} SNR.}
	\label{fig:LOSsensors}
\end{figure}
As expected, Fig.~\ref{fig:LOSsensors} shows that a larger number of LOS sensors results in better localization accuracy. For the cases where there is only one or two LOS sensor, the probability of correct recovery drops drastically because, in general, the minimum number of LOS sensors required for unambiguous TOA-based localization is three.

\subsection{Multiple Sources} \label{sub:multiple_sources}

In this example is evaluated the probability of correct recovery of multiple sources emitting different signals overlapping in the time and frequency domain. The SNR is fixed at \SI{30}{\decibel}. The definition of the probability of correct recovery defined in \eqref{eq:probability_detection} was for a single source. In the case of multiple sources, we define the average probability of correct recovery
\begin{equation} \label{eq:probability_detection_average}
	P_{\text{av}} = \frac{1}{ZQ}\sum_{z=1}^{Z}\sum_{q=1}^{Q}\mathds{1}\left(|\mathbf{p}_q-\hat{\mathbf{p}}_q^{(z)}|< \zeta\right),
\end{equation}
where $\mathbf{p}_q$ is the true location of the $q$-th source, $\hat{\mathbf{p}}_q^{(z)}$ is its estimate, and $\zeta$ is the error set to $\zeta=\nicefrac{r}{3}$ where $r$ is the waveform's ranging resolution as defined in \eqref{eq:signal_resolution}.
\begin{figure}
	\centering
	\includegraphics[width=\columnwidth]{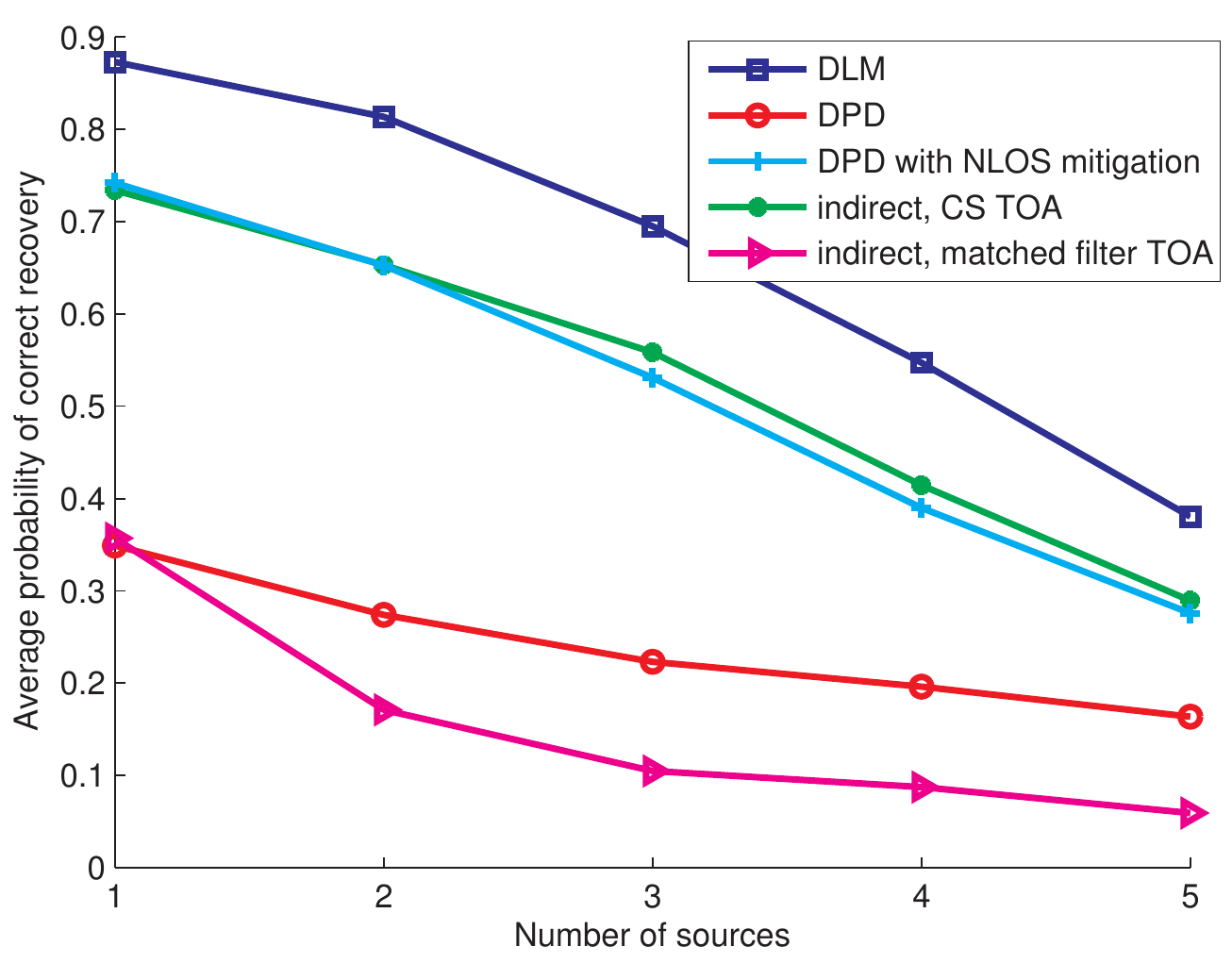}
	\caption{Average probability of correct recovery vs.\ the number of sources for a \SI{30}{\decibel} SNR.}
	\label{fig:sources_average}
\end{figure}
In Fig.~\ref{fig:sources_average}, it is shown how the average probability of correct recovery degrades as the number of sources increases. This is expected because the signals from different sources interfere with each other. Nonetheless, we can observe that DLM outperforms all other localization techniques when localizing multiple sources.

\section{Conclusions} \label{sec:conclusions}

By combining concepts from compressive sensing and direct localization, we have developed a  novel direct localization technique for mutliple sources in the presence of multipath (DLM).  This technique assumes the emitted waveforms are known but requires no prior information on the channel.  In fact, its localization accuracy is almost constant irrespective of the delay spread of the channel. At the core of our technique lies an optimization problem that recovers the locations of the sources with high accuracy by  exploiting properties that are different for LOS and NLOS paths. It is shown theoretically how to set the algorithm's parameters to guarantee successful recovery including a parameter that determines the relative contributions of the LOS and NLOS components to the cost function. Contrary to indirect techniques, the proposed technique is capable of localizing sources with an accuracy beyond that of the signal resolution, with high probability.
In absence of multipath, DLM's accuracy matches that of the maximum likelihood estimator of the sources' locations. In the presence of multipath, DLM's accuracy is better than the accuracy of all compared methods and can find the sources' location even when some sensors suffer from LOS blockage.
The gain in localization accuracy does not come for free, as DLM requires larger computational resources than previous techniques. To this end, we propose a grid refinement procedure which substantially reduces the computational complexity. Nonetheless, this should be less of a burden as computational power keeps increasing and second-order cone program solvers become more efficient. DLM's high accuracy is validated by extensive numerical simulations.


\appendices

\section{Proof of Lemma~\ref{lem:right_inequality}} \label{appen:right_inequality}

	Let an atomic decomposition  of $\mathbf{R}$ be \eqref{eq:atomic_decomposition}. The goal of the proof is to show that all locations, $\mathbf{p}_q^{(k)}$ for $q=1,\ldots,Q$ and $k=1,\ldots,K_q$ are consistent with $S_q$ or more paths if
	\begin{equation} \label{eq:right_inequality_1}
		\left\lVert\mathbf{p}_q^{(k)}\right\rVert_2 = u_q < \frac{1}{\sqrt{S_q-1}}.
	\end{equation}
	From \eqref{eq:atomic_decomposition}, the signal at the $l$-th sensor is
	\begin{multline} \label{eq:right_inequality_2}
		\mathbf{r}_l = \sum_{q=1}^{Q}\sum_{\substack{k=1\\b_q^{(k)}(l)\neq0}}^{K_q}c_{q}^{(k)}
		{b}_q^{(k)}(l)
		\mathbf{s}_q\left(\tau_l\left(\mathbf{p}_q^{(k)}\right)\right)
		+\\+
		\sum_{q=1}^{Q}\sum_{k=1}^{K_{ql}}
		c_{ql}^{(k)} e^{i\phi_{ql}^{(k)}}
		\mathbf{s}_q\left(\tau_{ql}^{(k)}\right).
	\end{multline}
	By Assumption~\ref{ass:source_identifiability}, $\tau_l(\mathbf{p}_q^{(k)})$ is a true propagation if $b_q^{(k)}(l)\neq0$.
	Therefore, if $\mathbf{b}_q^{(k)}$ has $S_q$ or more non-zero entries, 
	according to Definition~\ref{def:consistency}, $\mathbf{p}_q^{(k)}$ is consistent with $S_q$ or more paths. It is left to prove that $\|\mathbf{b}_q^{(k)}\|_0\geq S_q$. The proof is by contradiction.
	For instance,
	\begin{equation} \label{eq:right_inequality_4}
		\left\|\mathbf{b}_1^{(1)}\right\|_0 < S_1.
	\end{equation}
	and let the atomic decomposition \eqref{eq:atomic_decomposition} in which the atom $\mathbf{L}_1\left(\mathbf{b}_1^{(1)},\mathbf{p}_1^{(1)}\right)$ is replaced by $\|\mathbf{b}_q^{(k)}\|_0$ NLOS atoms as follows
	\begin{equation} \label{eq:right_inequality_5}
		\mathbf{L}_1\left(\mathbf{b}_1^{(1)},\mathbf{p}_1^{(1)}\right) =
		\sum_{\substack{l=1\\b_1^{(1)}(l)\neq 0}}^L \left|b_1^{(1)}(l)\right|
		\mathbf{N}_{11}\left(\tau_l\left(\mathbf{p}_1^{(1)}\right)\right).
	\end{equation}
	Consider now the two decompositions \eqref{eq:atomic_decomposition} and the one obtained with \eqref{eq:right_inequality_5}.
	The costs of the two decompositions differ only in the coefficients of the atoms shown in \eqref{eq:right_inequality_5}. Ignoring the common atoms, the cost of decomposition \eqref{eq:atomic_decomposition} is $c_{1}^{(1)}$, whereas the cost of decomposition obtained from combining \eqref{eq:right_inequality_5} with \eqref{eq:atomic_decomposition} is
	\begin{equation} \label{eq:right_inequality_7}
		c_{1}^{(1)}\sum_{\substack{l=1\\b_1^{(1)}(l)\neq 0}}^L \left|b_1^{(1)}(l)\right|.
	\end{equation}
	Normalizing the two costs by $c_{1}^{(1)}$, and if \eqref{eq:atomic_decomposition}, which by \eqref{eq:right_inequality_4} has a location $\mathbf{p}_1^{(1)}$ with less than $S_q$ paths, is optimal, then 
	\begin{equation} \label{eq:right_inequality_8}
		1 \leq \sum_{\substack{l=1\\b_1^{(1)}(l)\neq 0}}^L \left|b_1^{(1)}(l)\right|.
	\end{equation}
	
	We show next that inequality \eqref{eq:right_inequality_8} cannot be satisfied if $\|\mathbf{b}_1^{(1)}\|_2$ satisfies \eqref{eq:right_inequality_1}.	
	Define the vector function $\boldsymbol{1}(\mathbf{b}_{1}^{(1)})$ whose $l$-th entry is one if $b_{1}^{(1)}(l)\neq0$, and 0 otherwise, and denote $|\cdot|$ the element-wise absolute value. Then the right hand side of \eqref{eq:right_inequality_8} is
	\begin{equation} \label{eq:right_inequality_9}
		\sum_{\substack{l=1\\b_1^{(1)}(l)\neq 0}}^L \left|b_1^{(1)}(l)\right|  =
		\left[\boldsymbol{1}\left(\mathbf{b}_{1}^{(1)}\right)\right]^T \left|\mathbf{b}_{1}^{(1)}\right|,
	\end{equation} 
	and by the Cauchy-Schwarz inequality
	\begin{multline} \label{eq:right_inequality_10}
		\left[\boldsymbol{1}\left(\mathbf{b}_{1}^{(1)}\right)\right]^T \left|\mathbf{b}_{1}^{(1)}\right|
		\leq\\\leq
		\left\|\boldsymbol{1}\left(\mathbf{b}_{1}^{(1)}\right)\right\|_2 \left\|\mathbf{b}_{1}^{(1)}\right\|_2=
		\sqrt{\left\|\mathbf{b}_{1}^{(1)}\right\|_0} \left\|\mathbf{b}_{1}^{(1)}\right\|_2.
	\end{multline}
	However, $\|\mathbf{b}_{1}^{(1)}\|_2=u_{1}$, and by equation \eqref{eq:right_inequality_1}, $\|\mathbf{b}_{1}^{(1)}\|_2<\nicefrac{1}{\sqrt{S_1-1}}$. Moreover, by assumption \eqref{eq:right_inequality_4}, $\|\mathbf{b}_{1}^{(1)}\|_0\leq S_1-1$. Therefore, it follows
	\begin{equation}
		\sqrt{\left\|\mathbf{b}_{1}^{(1)}\right\|_0} \left\|\mathbf{b}_{1}^{(1)}\right\|_2<1,
	\end{equation}
	which combined with  \eqref{eq:right_inequality_9} and \eqref{eq:right_inequality_10} results in
	\begin{equation}
		\sum_{\substack{l=1\\b_1^{(1)}(l)\neq 0}}^L \left|b_1^{(1)}(l)\right|<1,
	\end{equation}
	which contradicts \eqref{eq:right_inequality_8}.

\section{Proof of Lemma~\ref{lem:left_inequality}} \label{appen:left_inequality}

Let \eqref{eq:atomic_decomposition} be an atomic decomposition of $\mathbf{R}$. 
Recall that parameter $K_q$ is the number of locations associated to the optimal atomic decomposition for the $q$-th source.
We aim to prove that if parameter $u_q$
\begin{equation} \label{eq:left_inequality_1st}
	\left\lVert\mathbf{b}_q^{(k)}\right\rVert_2 = u_q > \frac{1}{\sqrt{S_q}},
\end{equation}
then the optimal decomposition has $K_q\geq1$. The proof is by contradiction.
	 Let $K_1=0$, then a presumed optimal atomic decomposition \eqref{eq:atomic_decomposition} simplifies to
	\begin{equation} \label{eq:left_inequality_sufficiency_1}
		\mathbf{R}= \sum_{q=2}^{Q}
		\sum_{k=1}^{K_q}
		c_{q}^{(k)}\mathbf{L}_q\left(\mathbf{b}_q^{(k)},\mathbf{p}_q^{(k)}\right)
		+
		\sum_{q=1}^{Q}\sum_{l=1}^{L}\sum_{k=1}^{K_{ql}}
		c_{ql}^{(k)}\mathbf{N}_{ql}\left(\tau_{ql}^{(k)}\right).
	\end{equation}
	From \eqref{eq:left_inequality_sufficiency_1}, the signal at the $l$-th sensor is
	\begin{multline}
		\mathbf{r}_l = \sum_{q=2}^{Q}\sum_{\substack{k=1\\b_q^{(k)}(l)\neq0}}^{K_q}c_{q}^{(k)}
		{b}_q^{(k)}(l)
		\mathbf{s}_q\left(\tau_l\left(\mathbf{p}_q^{(k)}\right)\right)
		+\\+
		\sum_{q=1}^{Q}\sum_{k=1}^{K_{ql}}
		c_{ql}^{(k)} e^{i\phi_{ql}^{(k)}}
		\mathbf{s}_q\left(\tau_{ql}^{(k)}\right).
	\end{multline}
	Notice that the first summation begins with $q=2$, because $K_1=0$.
	By Assumption~\ref{ass:source_identifiability},
	\begin{equation} \label{eq:left_inequality_sufficiency_2}
		\left\{\tau_{1l}^{(k)}\right\}_{k=1}^{K_{1l}}
	\end{equation}
	are the true propagation delays of the paths between source $1$ and sensor $l$.
	By Assumption~\ref{ass:Sq_known}, there are $S_1$ LOS paths from source 1. Let 
	\begin{equation} \label{eq:left_inequality_sufficiency_3}
		\left\{l_1,\ldots,l_{S_1}\right\}
		\subseteq
		\left\{1,\ldots,L\right\}
	\end{equation}
	be the indexes of the destination sensors of such LOS paths, and let
	$\tau_{1l}^{(1)}$ in \eqref{eq:left_inequality_sufficiency_2} be the propagation delay corresponding to the LOS path between source 1 and sensor $l$, i.e.,
	\begin{equation} \label{eq:left_inequality_sufficiency_5}
		\tau_{1l}^{(1)}=\tau_l\left(\mathbf{p}_1\right) \text{ for } l\in\left\{l_1,\ldots,l_{S_1}\right\}.
	\end{equation}
	We show next that there exists a decomposition different than \eqref{eq:left_inequality_sufficiency_1} for which $K_1\geq1$ and whose cost is lower, thus contradicting the assumption that \eqref{eq:left_inequality_sufficiency_1} is optimal. 
	According to \eqref{eq:atoms_LOS} and \eqref{eq:atoms_NLOS}, the sum of NLOS atoms with delays $\tau_{1l}^{(1)}$ for $l\in\{l_1,\ldots,l_{S_1}\}$
	in the presumed optimal atomic decomposition \eqref{eq:left_inequality_sufficiency_1}, i.e., $\sum_{l\in\{l_1,\ldots,l_{S_1}\}} 	c_{1l}^{(1)} \mathbf{N}_{1l}(\tau_{1l}^{(1)})$, can be expressed for any parameter $c$ as
	\begin{multline} \label{eq:left_inequality_sufficiency_6}
		\sum_{l\in\{l_1,\ldots,l_{S_1}\}}
		c_{1l}^{(1)}
		\mathbf{N}_{1l}\left(\tau_{1l}^{(1)}\right)=
		\frac{\sqrt{S_{1}}c}{u_{1}}\mathbf{L}_{1}\left(\mathbf{b},\mathbf{p}_{1}\right)
		+\\+
		\sum_{l\in\{l_1,\ldots,l_{S_1}\}}
		\left(c_{1l}^{(1)}-c\right)
		\mathbf{N}_{1l}\left(\tau_{1l}^{(1)}\right),
	\end{multline}
	where $\mathbf{b}$ is
	\begin{equation}
		b(l) =
		\begin{cases}
			\frac{u_{1}}{\sqrt{S_1}}e^{i\phi_{1l}^{(1)}}
			& \text{for }l\in\left\{l_1,\cdots,l_{S_1}\right\} \\
			0
			& \text{otherwise}.
		\end{cases} \label{eq:left_inequality_sufficiency_7}
	\end{equation}
	Let $c=c_{\min}$ defined by
	\begin{equation}
		c_{\min} = \min_{l\in\{l_1,\ldots,l_{S_1}\}}
		c_{1l}^{(1)}. \label{eq:left_inequality_sufficiency_8}
	\end{equation}
	Next it is shown that the cost of the decomposition obtained by combining \eqref{eq:left_inequality_sufficiency_6}--\eqref{eq:left_inequality_sufficiency_8} with \eqref{eq:left_inequality_sufficiency_1} is lower than the cost of the decomposition \eqref{eq:left_inequality_sufficiency_1}, contradicting the assumption that \eqref{eq:left_inequality_sufficiency_1} is optimal. Notice the former decomposition includes the LOS atom $\mathbf{L}_{1}(\mathbf{b},\mathbf{p}_{1})$.
	The costs of the two decompositions differ only in the coefficients of the atoms shown in \eqref{eq:left_inequality_sufficiency_6}. Ignoring the common atoms, the cost of decomposition \eqref{eq:left_inequality_sufficiency_1} is
	\begin{equation}
		\sum_{l\in\{l_1,\ldots,l_{S_1}\}}
		c_{1l}^{(1)}
	\end{equation}
	whereas the cost of the decomposition obtained from \eqref{eq:left_inequality_sufficiency_6}--\eqref{eq:left_inequality_sufficiency_8} is
	\begin{equation}
		\frac{\sqrt{S_{1}}c_{\min}}{u_{1}}+
		\sum_{l\in\{l_1,\ldots,l_{S_1}\}}
		\left(c_{1l}^{(1)}-c_{\min}\right).
	\end{equation}
	Since \eqref{eq:left_inequality_sufficiency_1} is presumed optimal, it means it must
	satisfy
	\begin{equation}
		\sum_{l\in\{l_1,\ldots,l_{S_1}\}}
		c_{1l}^{(1)}\leq
		\frac{\sqrt{S_{1}}c_{\min}}{u_{1}}+
		\sum_{l\in\{l_1,\ldots,l_{S_1}\}}
		\left(c_{1l}^{(1)}-c_{\min}\right),
	\end{equation}
	which after simplification leads to $u_1\leq\nicefrac{1}{\sqrt{S_{1}}}$,
	contradicting \eqref{eq:left_inequality_1st}.

\bibliographystyle{IEEEtran}
\bibliography{IEEEabrv,references}

\end{document}